\documentclass[conference]{IEEEtran}
\makeatletter
\def\ps@headings{%
\def\@oddhead{\mbox{}\scriptsize\rightmark \hfil \thepage}%
\def\@evenhead{\scriptsize\thepage \hfil \leftmark\mbox{}}%
\def\@oddfoot{}%
\def\@evenfoot{}}
\makeatother
  \usepackage{graphicx}
  \usepackage{subfigure}
  \graphicspath{{./fig/}{../jpeg/}}

\usepackage{amsmath}
\usepackage{amssymb}
\usepackage{amsmath}
\usepackage{amsfonts}
\usepackage{color}
\newtheorem{definition}{Definition}
\newtheorem{theorem}{Theorem}
\newtheorem{lemma}{Lemma}

\usepackage{algorithmic}
\usepackage{algorithm}
\begin{document}
%
\title{Locating Multiple Ultrasound Targets in Chorus}

\author{\IEEEauthorblockN{Lei Song}
\IEEEauthorblockA{Institute for Interdisciplinary Information Sciences, \\ Tsinghua University, Beijing, P.R.China \\
Email:leisong03@gmail.com}
\and
\IEEEauthorblockN{Yongcai Wang}
\IEEEauthorblockA{Institute for Interdisciplinary Information Sciences, \\ Tsinghua University, Beijing, P.R.China \\
Email: wangyc@tsinghua.edu.cn}
}

\maketitle
\begin{abstract}
Ranging by Time of Arrival (TOA) of Narrow-band ultrasound (NBU) has been widely used by many locating systems for its characteristics of low cost and high accuracy.  However, because it is hard to support code division multiple access in narrowband signal, to track multiple targets, existing  NBU-based locating systems generally need to assign exclusive time slot to each target to avoid the signal conflicts.  Because the propagation speed of ultrasound is slow in air, dividing exclusive time slots on a single channel causes the location updating rate for each target rather low,  leading to unsatisfied tracking performances as the number of targets increases. In this paper, we investigated a new multiple target locating method using NBU, called \emph{UltraChorus}, which is to locate multiple targets while allowing them sending NBU signals simultaneously, i.e., in chorus mode. It can dramatically increase the location updating rate. In particular, we investigated by both experiments and theoretical analysis on the  necessary and sufficient conditions for resolving the conflicts of multiple NBU signals on a single channel, which is referred as \emph{the conditions for chorus ranging and chorus locating}. To tackle the difficulty caused by the anonymity of the measured distances,  we further developed \emph{consistent position generation algorithm} and \emph{probabilistic particle filter algorithm} to label the distances by sources, to generate reasonable location estimations, and to disambiguate the motion trajectories of the multiple concurrent targets based on the  anonymous distance measurements.  Extensive evaluations by both simulation and testbed were carried out, which verified the effectiveness of our proposed theories and algorithms.  
\end{abstract}
\section{Introduction}
\begin{figure*}[t]
\centering
\subfigure[one target, one receiver]{\begin{minipage}[c]{0.12\textwidth}
    \includegraphics[width=\textwidth]{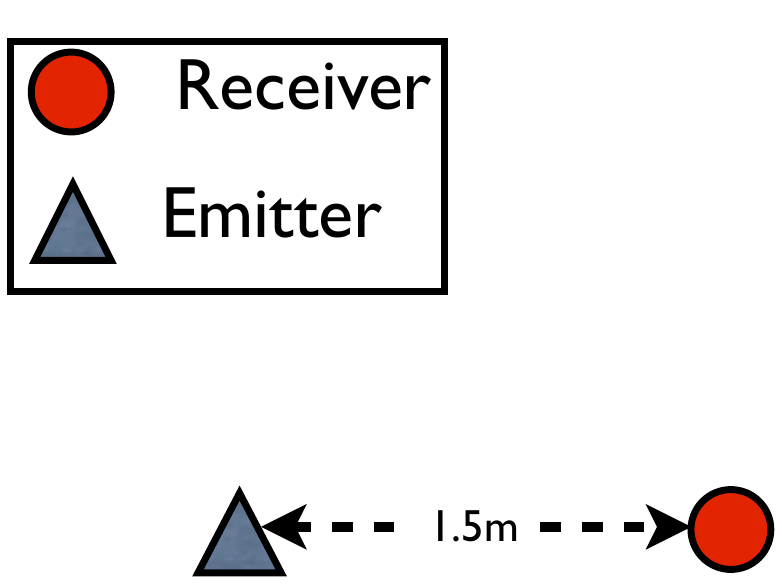}
\label{fig:1TO1OSC:a}
\end{minipage}
\begin{minipage}[c]{0.15\textwidth}
 \centering
 \includegraphics[width=\textwidth]{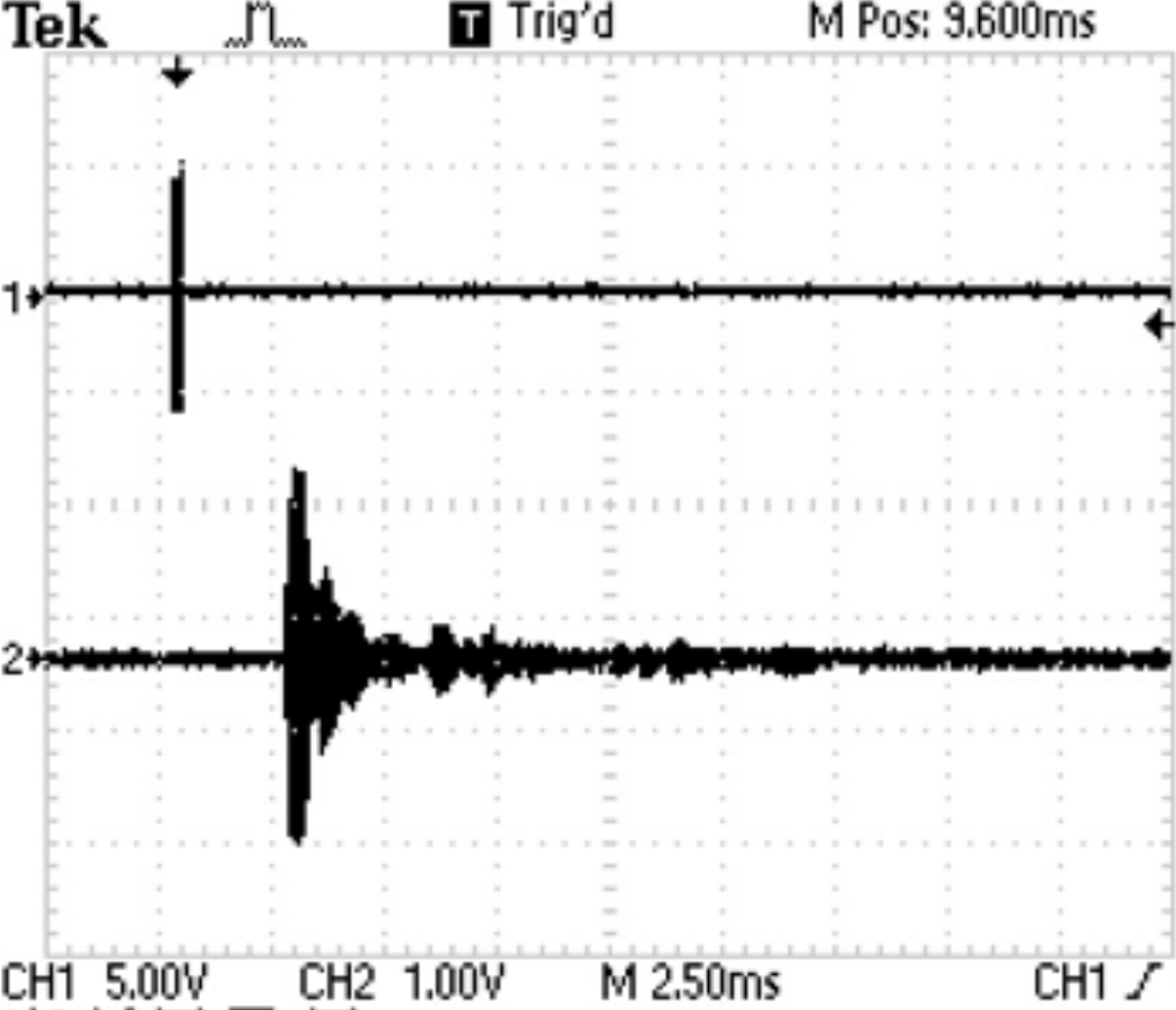}
\label{fig:1TO1OSC:b}
\end{minipage}
}
\hspace{0.03\textwidth}
\subfigure[two targets have different distances to the receiver]{\begin{minipage}[c]{0.12\textwidth}
 \centering
   \includegraphics[width=\textwidth]{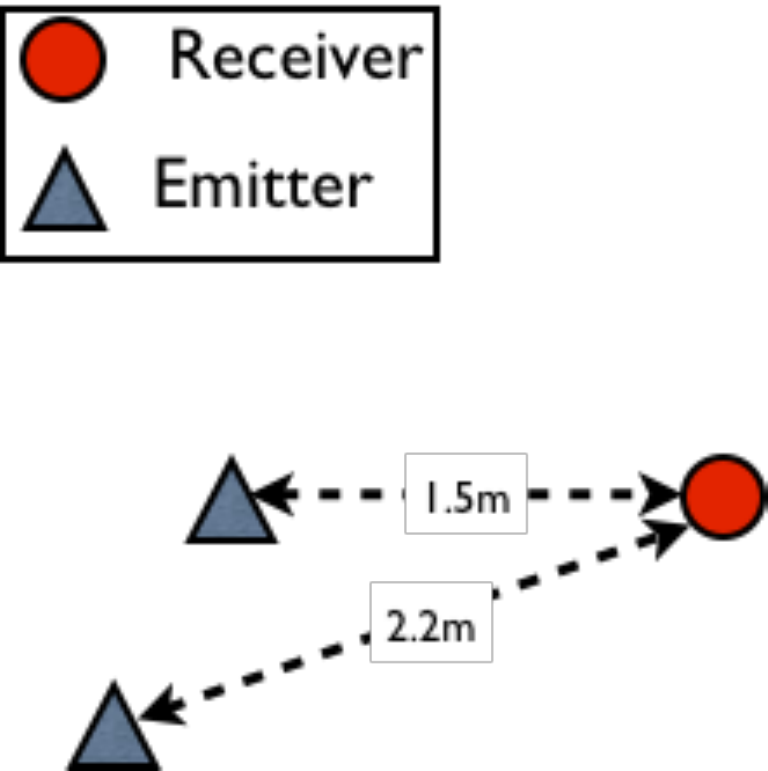}
\label{fig:2TO1OSC_CLR:a}
\end{minipage}%
\begin{minipage}[c]{0.15\textwidth}
    \centering
    \includegraphics[width=\textwidth]{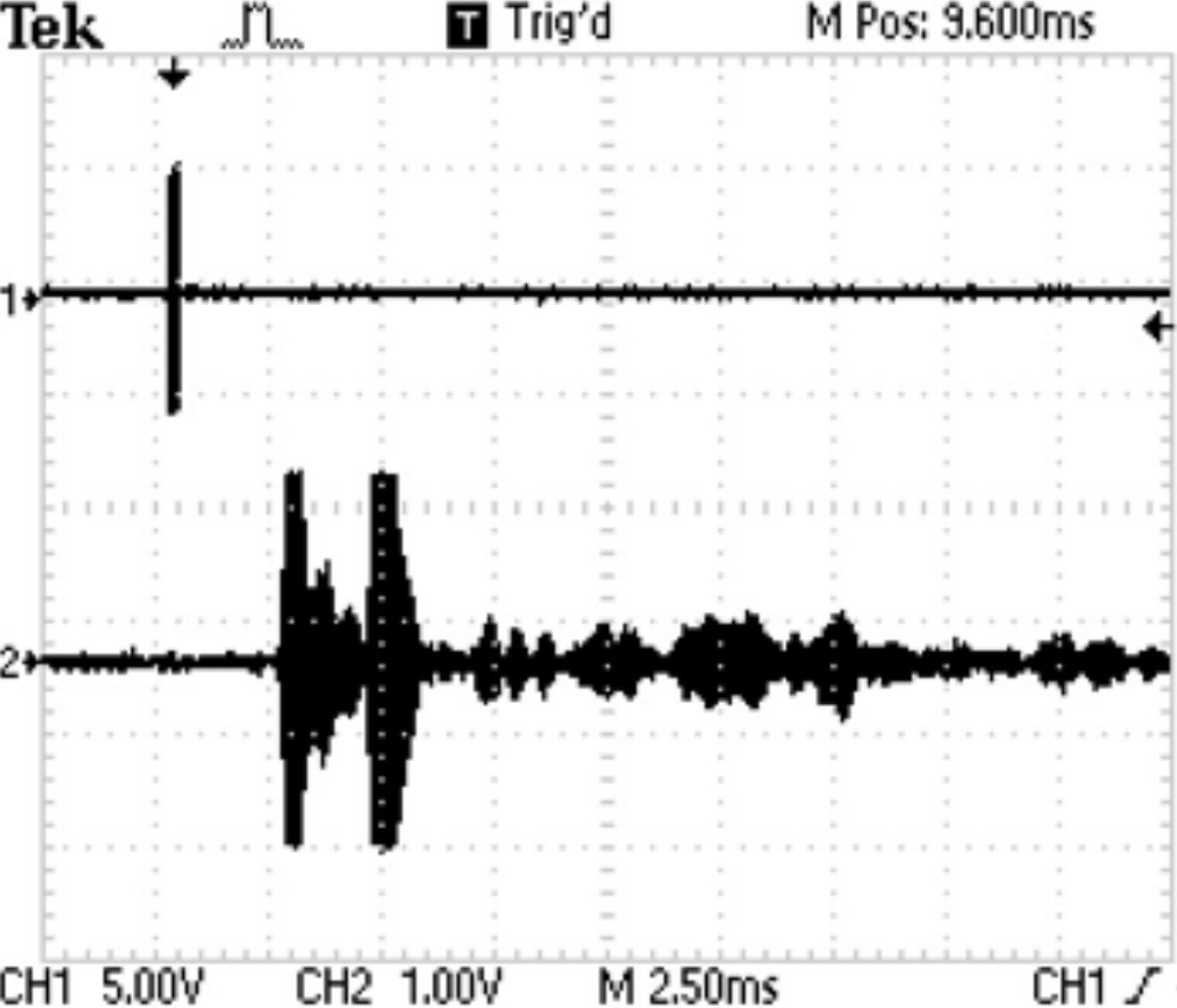}
\label{fig:2TO1OSC_CLR:b}
\end{minipage}}
\hspace{0.03\textwidth}
\subfigure[two targets have the same distance to the receiver]{\begin{minipage}[c]{0.12\textwidth}
    \centering
    \includegraphics[width=\textwidth]{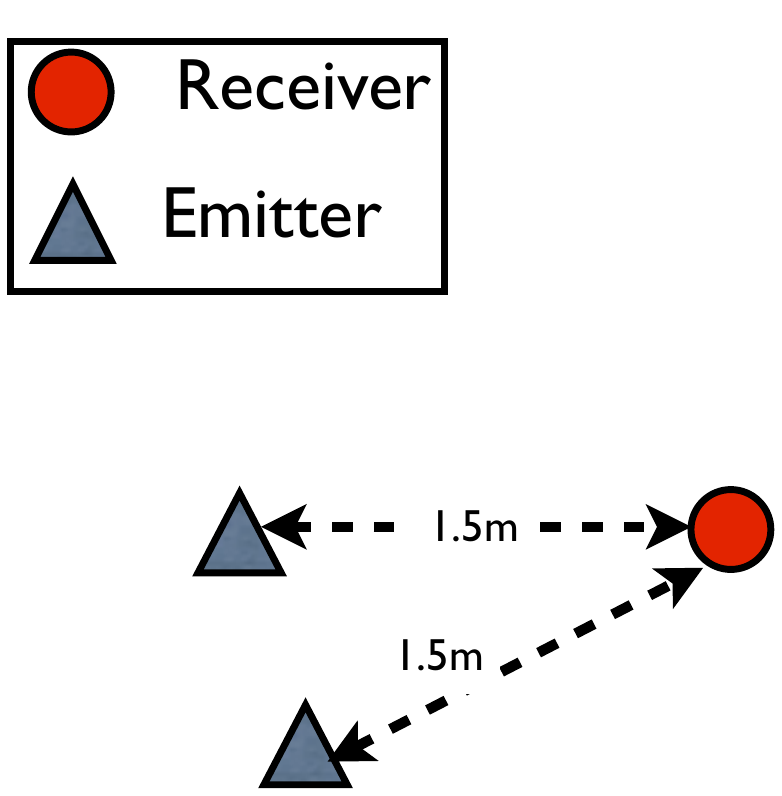}
\label{fig:2TO1OSC_OVR:a}
\end{minipage}%
\begin{minipage}[c]{0.15\textwidth}
    \centering
    \includegraphics[width=\textwidth]{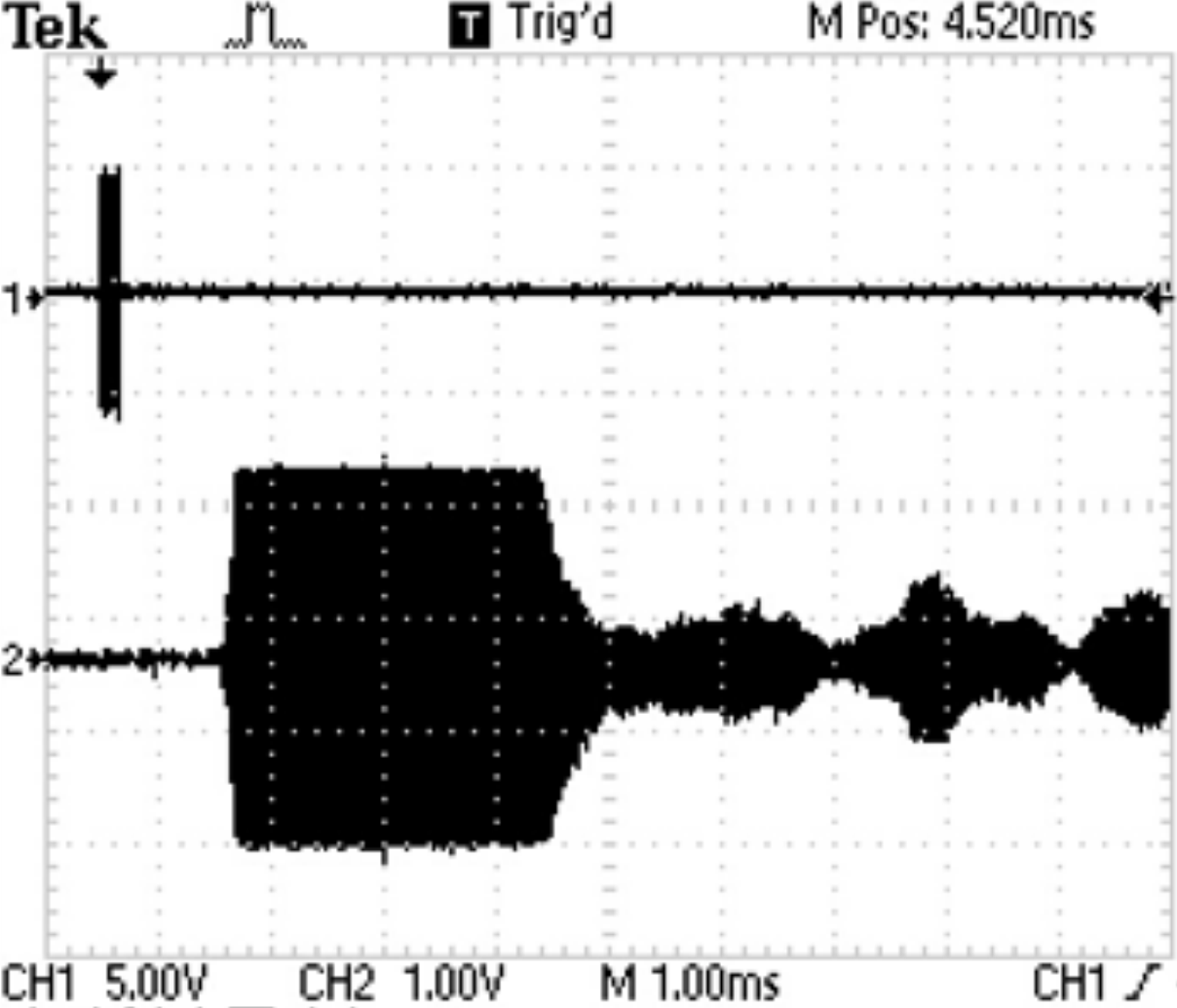}
    \label{fig:2TO1OSC_OVR:b}
\end{minipage}}
\caption{Experiments to test how the separation of NBU peaks are affected by the distances between the transmitters}
\end{figure*}

Locating by Time-of-Arrival (ToA)  of Narrow band ultrasound provides good positioning accuracy by using very low cost hardware and simple system architecture, which is widely
used in many locating system. e.g. ActiveBat\cite{Ward:1997ca}, Cricket\cite{Priyantha:2000hx}. However, when multiple targets are transmitting the same band ultrasound to a common set of receivers,  inevitable conflicts will happen at the receivers if the multiple targets are not appropriately coordinated. Further, because it is hard to code target ID into NBU, even if the NBU signals from multiple targets can be separated at a receiver, the receiver can hardly determine the source (transmitter) of the NBUs, resulting at locating ambiguity. To tackle these multiple target locating problems, existing NBU based locating systems generally rely on the exclusive working mode of the multiple targets, in which each target is assigned an exclusive time slot by TDMA or CSMA scheme to guarantee the NBU signal transmitted from one target is not conflicted to the others.

However, because the propagation speed of the ultrasound is rather slow in the air (e.g., 100 ms are needed for ultrasound propagating 34 meters), the time slot for each target's each transmission has to be long enough to avoid the transmitted NBU being conflicted with the same frequency NBUs from the other targets. Therefore, in exclusive mode\cite{AUITS}\cite{Priyantha:2000hx}\cite{Ward:1997ca}, at any time, only NBU from one target is in propagation, which results at low locating updating rate for individual target when the number of target is large. This on one hand limits the locating capacity (number of simultaneously located targets), on the other hand affects the tracking fidelity, especially when the targets are moving quickly.


To deal with these problems, in this paper, we investigated the problem of  locating multiple NBU targets in chorus mode, which is to locate a set of targets concurrently by allowing them to transmit NBU signals in the same time slot. 
In this study, we conducted not only theoretical analysis,  but also extensive simulations and hardware experiments.  
In particular, we addressed the difficulties of \emph{chorus ranging} (measuring TOAs from multiple concurrent targets) and \emph{chorus locating} (calculating locations for the multiple targets)  from following five aspects: 

1) We investigated via experiments on the conditions for a receiver to reliably separate the multiple NBUs from multiple concurrent targets.   
2) It leads to the geometric conditions on the relationship among the targets to guarantee  non-conflict multiple TOA measurements. 
 3) Since the measured TOAs lack source identity, we present consistent position generation algorithm, which exploits the historical consistence (in terms of the deviation to the historical position of the targets) to label the potential sources (transmitters) for the anonymous distances, and then to generate and to filter the potential positions via evaluating their self-consistence (in terms of the residue of location calculation).  
  4) By using the generated consistent positions as input, we proposed probabilistic particle filter algorithm to further disambiguate the trajectories of the multiple targets by using the consistence of the moving speeds and accelerations of targets as the evaluation metrics.  
 5) At last, location based transmission scheduling algorithm was proposed, which schedules the concurrent transmitters for reliable, online multi-target locating in chorus mode.

%
%
The rest of this paper is organized as follows. Related work and background are introduced in Section II. We introduced the feasibility of chorus ranging in Section III. The conditions for successful multi-target chorus locating are presented in section IV. Techniques for identifying potential sources of TOA and the particle filter algorithms for trajectory disambiguating  are presented in Section V. We proposed location-based transmitter scheduling scheme in Section VI. Simulations and experimental results are presented in Section VII. The paper is concluded with remarks in Section VIII. 

\section{Related Work and Background}
Ranging by TOA of NBU is a very attractive technique for fine-grained indoor locating due to its high accuracy, low cost, safe-to-user and user-imperceptibility. It can provide positioning accuracy in centimeter level even in 3D space, which makes it very fascinating in may indoor applications. Popular ultrasound TOA-based indoor locating system include Bat\cite{Ward:1997ca}, Cricket \cite{Priyantha:2000hx}, AUITS\cite{AUITS}, LOSNUS\cite{LOSNUS}, etc. Popular application scenario include location-based access control \cite{lock}, location based advertising delivery\cite{ubicomm09}, healthcare etc.    

Multiple target locating problem has been investigated in existing systems. When using NBU for TOA-based ranging, there is no room for coding the ID of target. Existing approaches let the target send a ultrasound-radio frequency (RF) pair. The RF signal is for synchronization and identification\cite{Ward:1997ca}\cite{Priyantha:2000hx}\cite{AUITS}.  Since there are several Media Access Control(MAC) protocol for RF signal, they can be adopted to coordinating target by just extending the length of the time-slot. 

Another approach is to explore the broadband ultrasound. Compared to the narrowband version, broadband ultrasound requires the  transducer  \cite{Hazas:2006tv} to have better frequency response performance. The broadband ultrasound wave can accommodate identity of target to support multiple targets. Furthermore, if the wave is encoded with orthogonal code\cite{Momo:2010ty}, two waves can be
decoded respectively even overlapped. But broadband locating needs high cost transducers, and the signal is more sensitive to the Doppler effects. To the best of our knowledge, very few results have been reported for locating in chorus mode, because the collision problem of NBUs are generally hard to tackle. In this paper, we investigate conditions and algorithms to resolve this challenge.

\section{Feasibility of Chorus Ranging}
At first, we introduce \emph{exclusive mode} and \emph{chorus mode} and presents experiments to investigate the conditions for a receiver to  successfully detect NBUs from concurrently transmitting targets. 

\subsection{Exclusive Mode}

In conventional approach, when there are multiple targets, to avoid conflicting of NBUs, RF+US signals from different targets are scheduled into different time slots (called  \emph{exclusive mode}). In each slot one target broadcasts RF+US signals simultaneously, where the RF signal is used to synchronize timers among the target and the receivers. Then the synchronized receivers measure the TOAs of the successive ultrasound wave from the target to estimate their distances to the target and to calculate the target's location via by least square estimation or trilateration\cite{ilssurvey}. The exclusive slot assignment can be realized by utilizing the media access control(MAC) protocols of the RF signal, e.g., CSMA, TDMA\cite{Priyantha:2000hx}\cite{Ward:1997ca}\cite{AUITS}. 

But, because the propagation speed of ultrasound is quite slow in the air (340 m per second), the time-slot for each exclusive target has to be long enough to avoid NBU conflicting to the successively arrived NBU from other targets. For an example,  Cricket\cite{Priyantha:2000hx} assigns each target nearly $100ms$ by CSMA protocol. Because $n$ targets need $n$ exclusive slots, the location updating rate of each individual target is only $O(\frac{1}{n})$, which may become unsatisfactory when there are large number of targets.  

\subsection{Chorus Mode}
In contrast to the exclusive mode, in chorus mode, we allow multiple targets to broadcast NBUs in the same time slot. A general way is to use a \emph{RF commander} to broad RF to synchronize the timers of the targets and the receivers, and let the targets to broadcast NBU signals simultaneously and concurrently with the RF. Each receiver detects the mixed ultrasound signals from the multiple targets in its communication range and tries to separate the NBU signals to estimate the TOAs from the targets, and then to determine the locations of the multiple targets. 
\subsection{Experiments on Multiple NBU signal Detection}
Detecting TOAs from concurrently transmitted NBU waves at the receiver is the critically first step for chorus ranging, which determines the feasibility of chorus ranging and locating. We conducted experiments using MTS450CA Cricket nodes \cite{Priyantha:2000hx} to investigate the conditions for successfully multiple TOA detection at  a receiver.  

Before carrying out the experiments, we made some modification to  the firmware of Cricket node. Firstly, the policy  to detect only the first arising edge was canceled, which is originally designed in Cricket to filter out the NLOS (non-line-of-sight) and the echo signals, because the NLOS and echo waves arrive later than the direct path NBU. In the new version, the received wave power is continuously compared to a threshold. When a rising edge (or wavefront) is detected,  a TOA event is reported and the comparator state is set to ``high''. When the wave power decreases to be lower than the threshold, the comparator state returns to ``low'' to be ready for detecting the next wavefront.  Secondly, we disable the CSMA protocol in target, so that the targets can send ultrasound simultaneously.

\subsubsection{Aftershock}
The first experiment used one receiver and one target.  The screen-shot on oscilloscope is shown in Fig. \ref{fig:1TO1OSC:b}.
The target send $200\mu s$ ultrasound wave. After about $4$ms, this NBU wave arrives at the receiver, which cause a $1$ms shock on the receiver's sensor. Because the ultrasound is mechanical wave, the shock on the receiver is much longer than the length of the wave sent by the target. This phenomenon is called \emph{aftershock}. When the sensor in the receiver is experiencing an aftershock, the comparator in the sensor is kept in \emph{high state},  which will block the detection of the newly arrived NBU wavefront.  In other word, aftershock will cause loss of TOA measurements at the receiver. Intuitively, the longer is the aftershock, the more frequent is the loss. From the oscilloscope output, we can also see some secondary peaks caused by the echoes. These secondary peaks can be filtered out because their powers are lower than the threshold. After the energy of the aftershock fades below the threshold, the comparator switches to \emph{low state}, which is ready for detecting the next NBU wave. 

\subsubsection{Multiple TOA Detections}Two targets and one receiver are used in the second experiment,  in which the two targets are placed at different distances from the receiver. When the two targets broadcast ultrasound signal simultaneously, the detected waves at the receiver are shown in Fig. \ref{fig:2TO1OSC_CLR:b}. 
In this case, the receiver detects two NBU wavefronts successfully (i.e., two TOAs of ultrasound), which is because the separation between the wavefronts is greater than the length of the aftershock, but note that the detected TOAs are anonymous, i.e., the receiver don't know their targets.  In the third experiments, the two targets have the same distance to the receiver, their generated waves at the receiver are overlapped, as shown in Fig. \ref{fig:2TO1OSC_OVR:a}.  In this case, only one TOA is measured at the first arising edge, which is also anonymous. 
\subsection{Condition on Detecting Multiple TOAs}
The above experiments showed clearly that  whether two successive NBU waves arrived at a receiver can be successfully detected is determined by the time separation between the two waves. If the time separation is longer than the length of the aftershock generated by the first wave, the wavefront of the second wave can be detected, otherwise, the second wavefront will be lost because the comparator is already in high state.  Since the length of the aftershock is affected by the received energy of the NBU signal detected at the receiver and by the inertia of the ultrasound transducer of the receiver, it will be better to choose ultrasound transducer with weak inertia to get shorter aftershock to improve the capability of detecting the successive ultrasound pulses.   To formulate the impact of the aftershock, let's denote $L_{max}$ as the longest possible aftershock generated by the strongest signal at the receivers. Let $v_{u}$ be the speed of the ultrasound, then 
\begin{definition}[confident separation distance]
We define $\omega=L_{max}v_{u}$ as the confident separation distance between the successively arrived waves for the receiver to successfully detect their TOAs.  
\end{definition}

From triangle inequality, it is easy to verify that if the TOAs from two concurrent targets can be successfully detected by a receiver, the distance between the two targets must be larger than $\omega$.  Let's further take the audible region of the ultrasound into consideration. We assume all the targets have the same broadcasting power, then:
\begin{definition}[audible range of ultrasound] 
We define $r$ as the audible range of the ultrasound, which is the propagation distance of the ultrasound from a target before the wave power is lower than the detectable threshold of the receivers. 
\end{definition}

By combining the separation distance and the audible range, we can arrive at the condition for a receiver to successfully detect TOAs from two concurrent targets. 
\begin{theorem}\emph{
We consider two targets $a$ and $b$ are at location $\mathbf{x}_{a}$ and $\mathbf{x}_{b}$ respectively, who send NBU waves in the same time slot, a receiver at location $\mathbf{x}_{x}$ can detect the TOAs of the two waves if:   
   \begin{equation}
\left\{ {\begin{array}{*{20}{c}}
  {\left| {{d_{a,x}} - {d_{b,x}}} \right| > \omega } \\ 
  {{d_{a,x}} \leqslant r,{d_{b,x}} \leqslant r} 
\end{array}} \right.
   \end{equation}
    \label{lemma:first}}
\end{theorem}
where $d_{i,j}$ calculates the distance between $\mathbf{x}_{i}$ and $\mathbf{x}_{j}$.  
\begin{figure}[htpb]
 \centering
   \includegraphics[width=.12\textwidth]{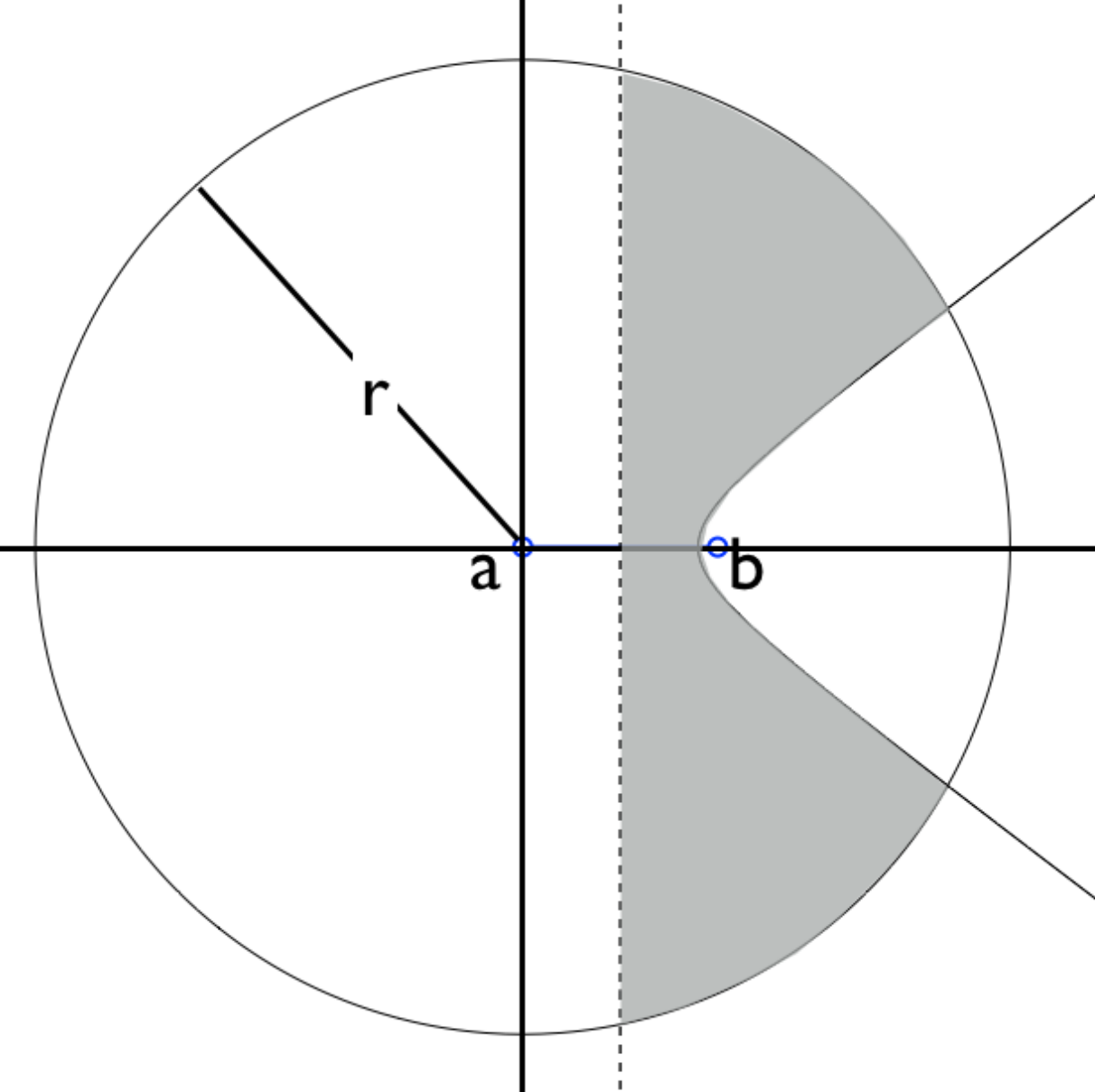}
    \caption{The blind region of $a$ caused by $b$ }
    \label{fig:BlindRegion}
\end{figure}

For the case of multiple targets, to check whether a receiver can detect TOAs from their concurrent transmissions, we can simply sort their distances to the receiver in an ascending order. When the difference between any two adjacent sorted distances are larger than $\omega$, and when the receiver is in their common audible region, the receiver can successfully detect the TOAs of their concurrent ultrasound waves. 
\subsection{``Blind Region'' Impacted by a Concurrent Target}
Based on Theorem\ref{lemma:first}, let's now consider in which region will a receiver lose the TOA from a target $a$ when another target $b$ is transmitting concurrently.  
\begin{definition}[Blind region]
Blind region of $a$ caused by $b$ is referred to the region in which the receivers cannot capture TOA from $a$, if $a$ and $b$ send wave at the same time. The area of blind region of $a$ caused by $b$ is denoted by $S^{B}_{a\leftarrow b}$. 
\end{definition}

 \begin{figure}
    \subfigure[$d_{a,b}>2r$]{
	\label{fig:mini:subfig:a} 
	\begin{minipage}[b]{0.12\textwidth}
	    \centering
	    \includegraphics[width=0.8in]{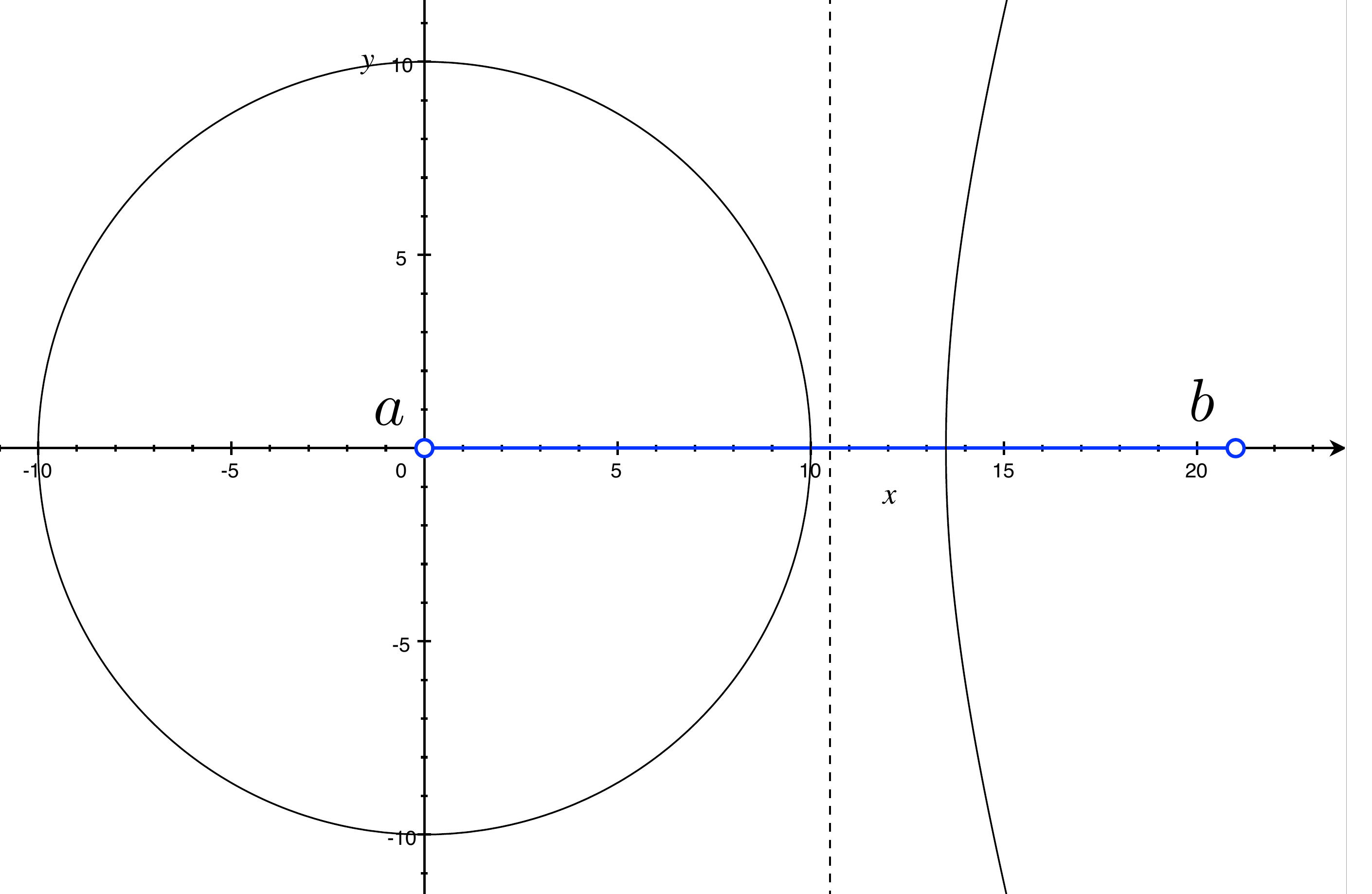}
	\end{minipage}}%
    \subfigure[$2r\!\! -\!\! \omega\!\! \le\!\!  d_{a,b}\!\! \le\!\!  2r$]{
	\label{fig:mini:subfig:b} 
	\begin{minipage}[b]{0.12\textwidth}
	    \centering
	    \includegraphics[width=0.8in]{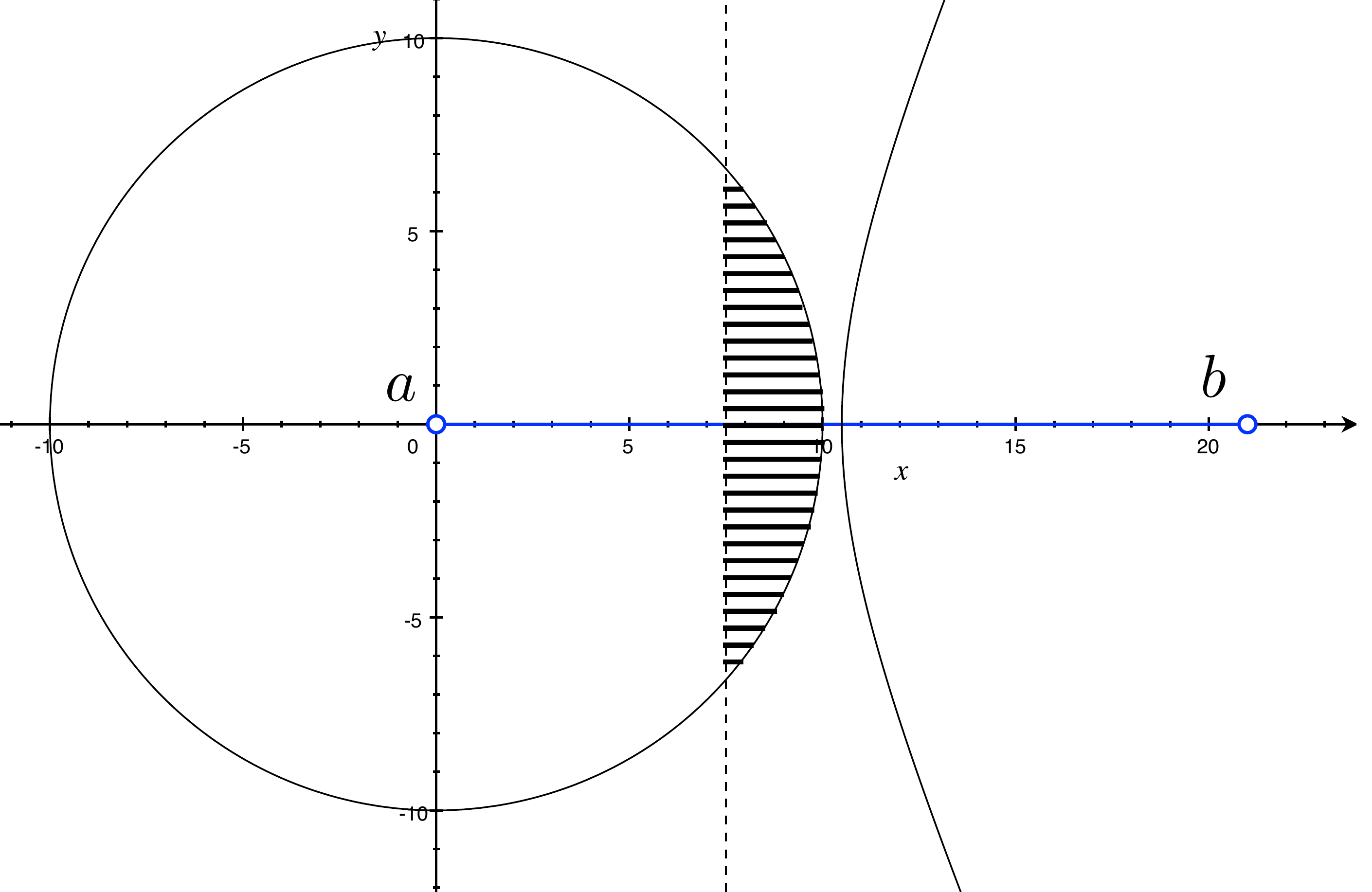}
	\end{minipage}}%
    \subfigure[$\omega\!\! <\!\!d_{a,b}\!\!< \!\!2r\!\!-\!\!\omega$]{
	\label{fig:mini:subfig:c} 
	\begin{minipage}[b]{0.12\textwidth}
	    \centering
	    \includegraphics[width=0.8in]{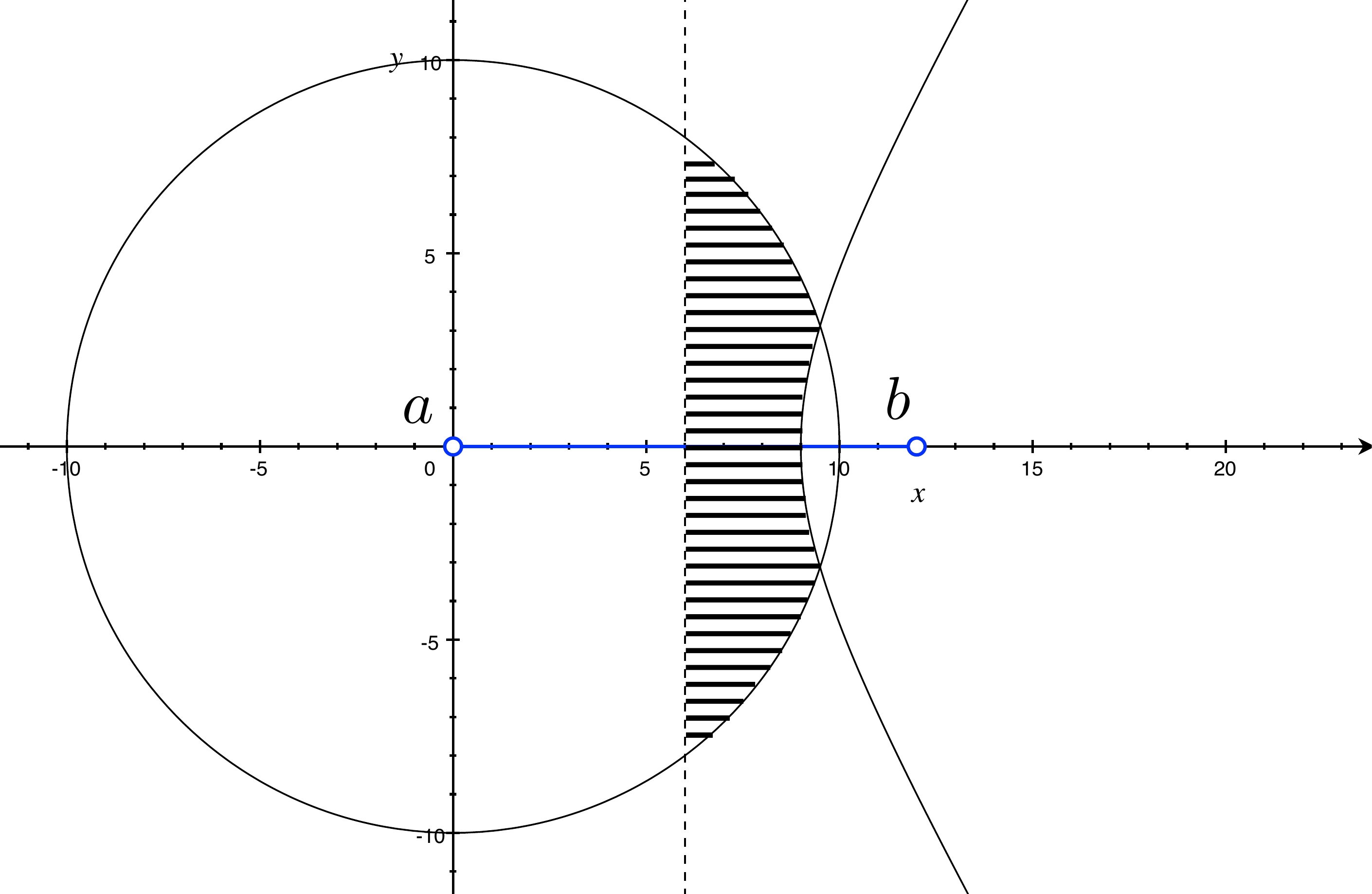}
	\end{minipage}}%
    \subfigure[$0\!\! <\!\! d_{a,b}\!\! \le\!\!  \omega $]{
	\label{fig:mini:subfig:d} 
	\begin{minipage}[b]{0.12\textwidth}
	    \centering
	    \includegraphics[width=0.8in]{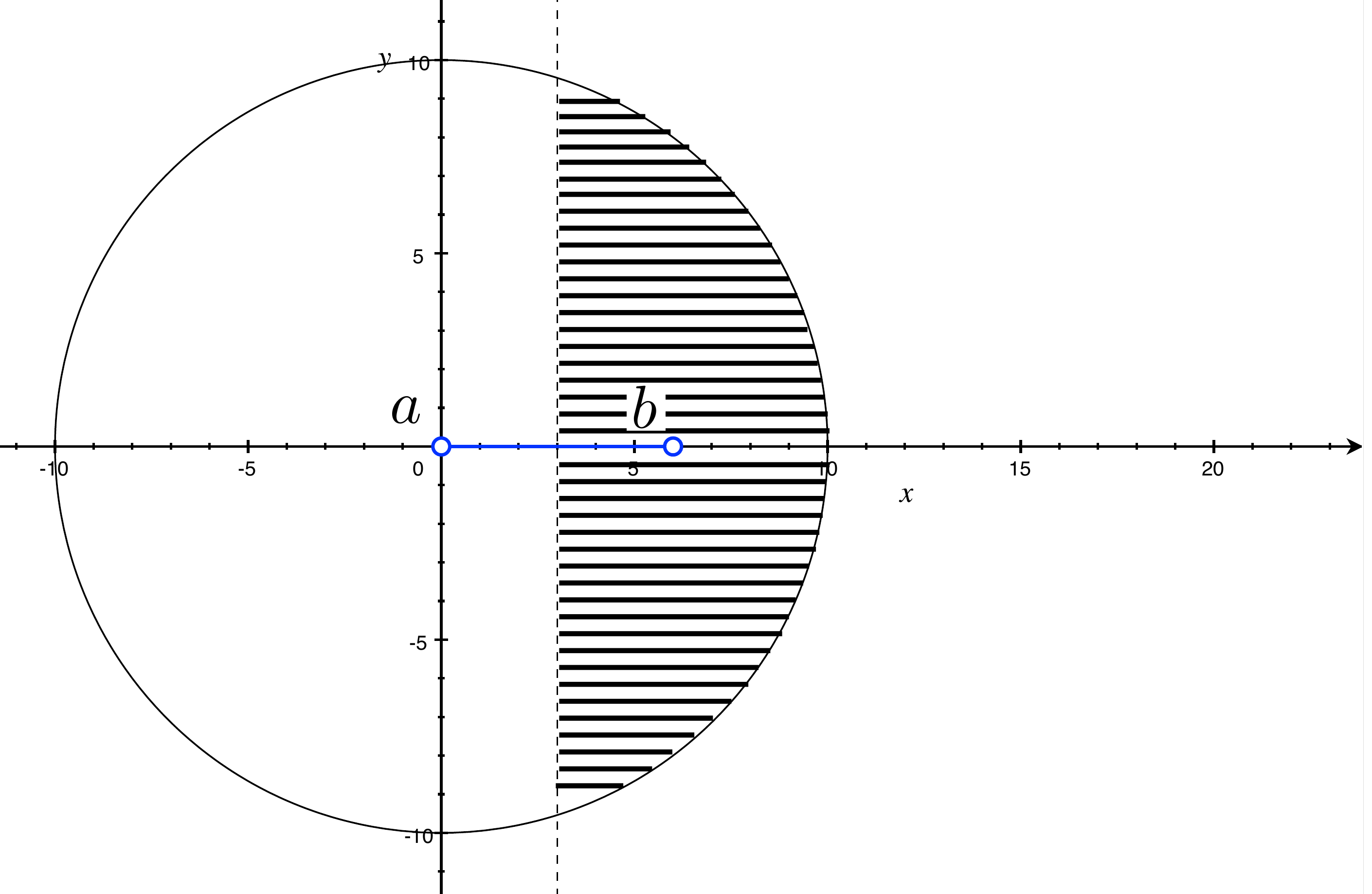}
	\end{minipage}}
    \caption{The grey region stands for blind region, where the rest part in audible-circle is
    audible region}
    \label{fig:f4} 
\end{figure}
\begin{figure*}
    \subfigure[Blind Region by 2 targets]{
	\label{fig:BObj2} 
	\begin{minipage}[b]{0.18\textwidth}
	    \centering
	    \includegraphics[width=1.2in]{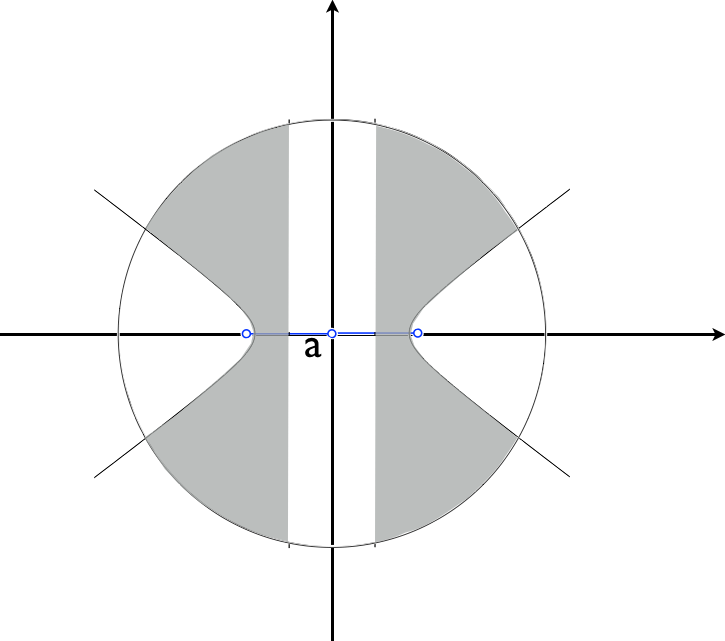}
	\end{minipage}}%
	\hspace{0.005\textwidth}
    \subfigure[by 3 targets]{
	\label{fig:BObj3} 
	\begin{minipage}[b]{0.18\textwidth}
	    \centering
	    \includegraphics[width=1.2in]{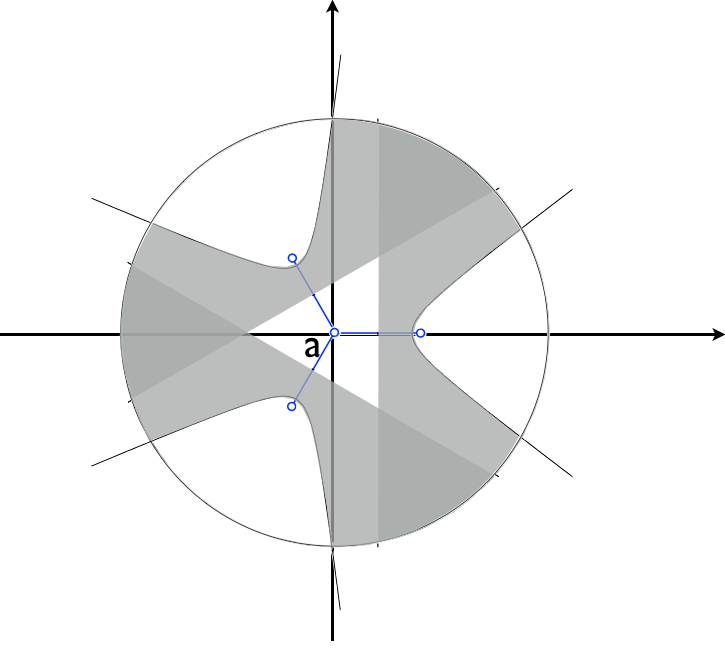}
	\end{minipage}}
	\hspace{0.005\textwidth}
    \subfigure[by 4 targets]{
	\label{fig:BObj4} 
	\begin{minipage}[b]{0.18\textwidth}
	    \centering
	    \includegraphics[width=1.2in]{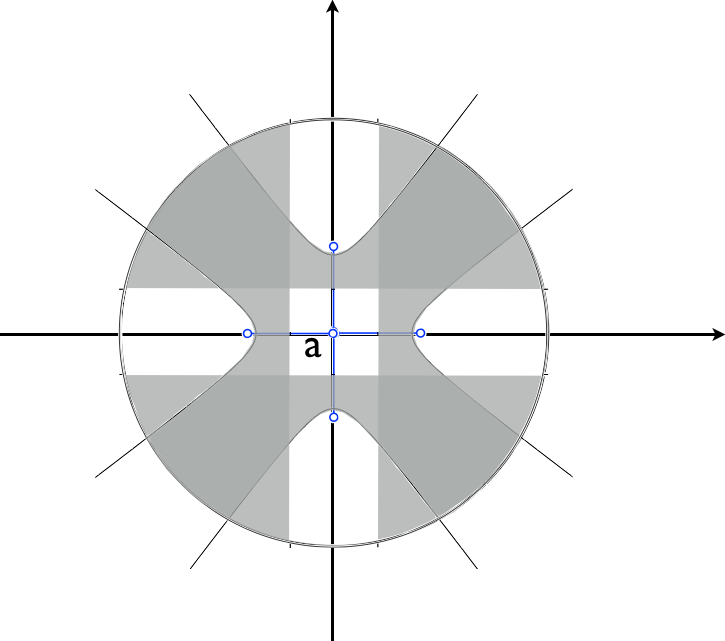}
	\end{minipage}}%
	\hspace{0.005\textwidth}
    \subfigure[by 5 targets]{
	\label{fig:BObj5} 
	\begin{minipage}[b]{0.18\textwidth}
	    \centering
	    \includegraphics[width=1.2in]{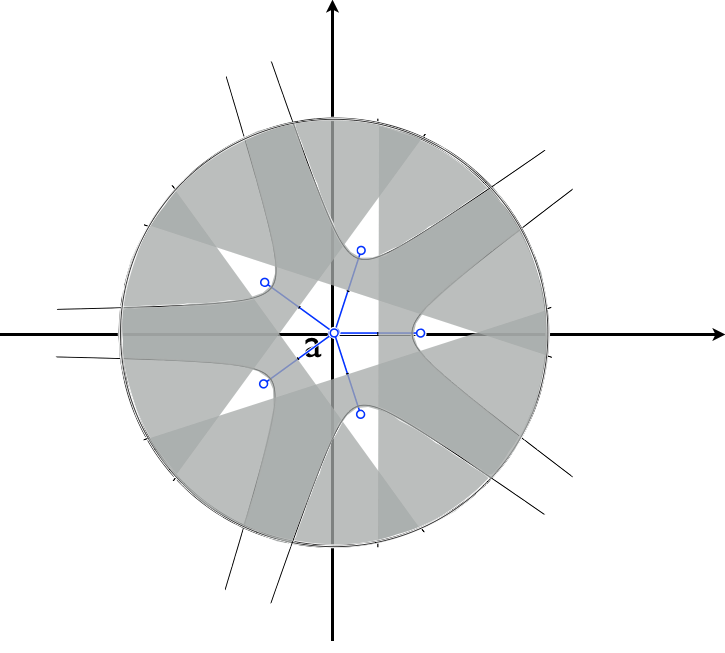}
	\end{minipage}}
	\hspace{0.005\textwidth}
    \subfigure[by 6 targets]{
	\label{fig:BObj6} 
	\begin{minipage}[b]{0.18\textwidth}
	    \centering
	    \includegraphics[width=1.2in]{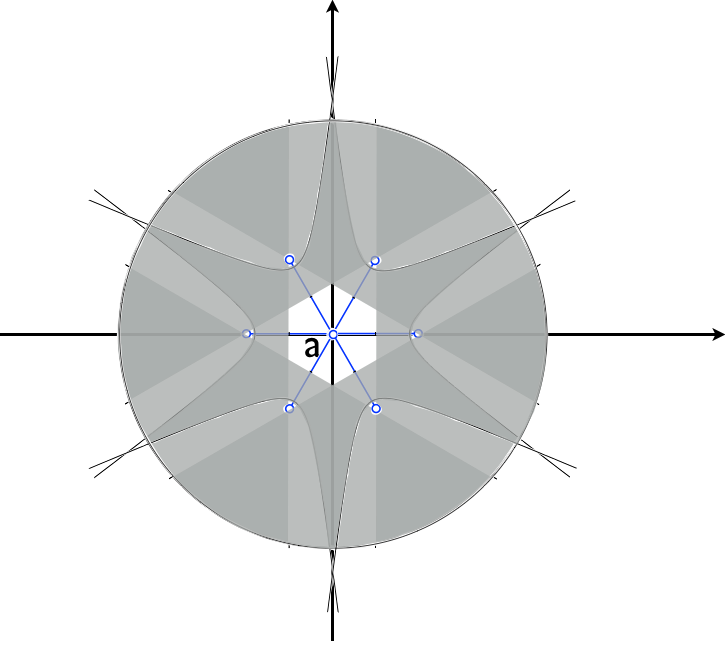}
	\end{minipage}}
    \caption{Blind-region of target $a$ caused by different number of other targets}
    \label{fig:2to6} 
\end{figure*}

In Fig.\ref{fig:BlindRegion} the blind region of $a$ caused by $b$ is shown by the gray region, which is characterized by inequality functions: $0<d_{a,x}-d_{b,x} \le \omega$ and $d_{ax}\le r$. When the receivers are locating in this region, the NBU wave from $a$ will be hidden in the aftershock of the wave from $b$, so that the receiver cannot detect TOA from $a$. Depending on the distance between $a$ and $b$, i.e., $d_{a,b}$, $S^B_{a\leftarrow b}$ changes
from 0 to $\frac{\pi r^2}{2}$. Figure \ref{fig:f4} shows how $S^B_{a\leftarrow b}$ changes with $d_{a,b}$, which indicates that $S^B_{a\leftarrow b}$ is a function of $d_{a,b}$. Moreover the area of blind-region can be expressed in close form. 
\begin{equation}
    S^B_{a\leftarrow b}(d_{a,b})\!\!=\!\!\left\{ \begin{array}{ll}
	0 & \textrm{$d_{a,b}>2r$}\\
	r^2(\theta-\sin \theta \cos \theta) & \textrm{$2r-\omega \le d_{a,b}\le 2r$}\\
	r^2(\theta-\sin \theta \cos \theta)-S_e &  \textrm{$\omega <d_{a,b}< 2r-\omega $}\\
	r^2(\theta-\sin \theta \cos \theta) &  \textrm{$0<d_{a,b}\le \omega $}\\
    \end{array} \right.
    \label{equ:blind}
\end{equation}
For clarity of expression, the detailed expansion of $S^B_{a\leftarrow b}(d_{a,b})$ can be referred in Appendix. We can just note that it is a monotone decreasing function of $d_{a,b}$. 
\section{Conditions for Chorus Locating}
Now let's consider the conditions for localizing multiple targets in the chorus mode. It is widely known that in ranging based locating algorithms such as trilateration,  three distance measurements from non-collinear beacons are necessarily required for uniquely determining the position of a target.  We therefore investigated the condition for obtaining at least three TOA measurements for a target in the chorus mode. Note that for randomly deployed receivers, the probability of three chosen receivers are collinear are small, therefore, the non-colinear constraint is not considered at this stage.    
\subsection{How Many TOA Detectable Regions Are Left?}
We define the \emph{TOA detectable region (TDR)} of a target as its audible region minus its blind region.  Fig.\ref{fig:eandb} shows the blind region caused by one concurrent target. The white region in the audible circle is the TDR region. When multiple concurrent targets are presenting, the left TDR will be further reduced. We denote the TDR of a target $a$ caused by a concurrent target set $T$ as $S^{D}_{a\leftarrow T}$. The area of $S^{D}_{a\leftarrow T}$ will affect the possible number of receivers in it for whatever distributions of the receivers, which determines the number of TOAs that can be obtained for a target. 
 \subsubsection{Consider Pairwise Separation Among Targets}
When the number of the concurrent targets  is more than 2, the blind region of the target $a$ is the union area of the blind-regions caused by all other targets in set $T$. 
\begin{equation}
    S^{B}_{a\leftarrow T}=\cup_{s\in T} S^{B}_{a\leftarrow s}
    \label{equ:BlindUnion}
\end{equation}
As indicated in (\ref{equ:blind}), $S^{B}_{a\leftarrow b}$ is a monotone decreasing function of $d_{a,b}$, therefore, an intuition is that the less are the pair-wise distances among the targets, the larger is the blind region cased by each target.  Therefore, we consider $S^{B}_{a\leftarrow T}$ when the pair-wise distances among all concurrent targets are the same, denoted by $d$. Via such a way, we characterize how the inter distances among the targets and their distributions affect the blind region of a particular target. 

 \subsubsection{Lower Bound of $S^{D}_{a\leftarrow T}$ in Multiple Target Case}
When all targets have the same pair-wise distance $d$, because the isotropous feature of the audible circle of $a$, the blind region caused by each individual target has the same shape and the same size. By inclusion-exclusion principle, the union area of these blind regions is the largest when the intersection area of the blind regions is the smallest. This case appears when the other targets are geometrically symmetrically distributed around $a$. Fig.\ref{fig:BObj6} shows the largest union area of the blind regions of $a$ when $\left| T \right| =2,3,4,5,6$ respectively. The corresponding TDR area is the lower bound of $S^{D}_{a\leftarrow T}$ for pairwise separation distance $\ge d$ and when the number of concurrent targets in the audible region is known. We omit the expressions of these lower bounds for space limitation.

More generally, when there are unknown number of targets are presenting,  we can also derive a lower bound of  $S^{D}_{a\leftarrow T}$ for given $d$. It is the area of the inscribed circle centered at $a$ with radius $d/2$ in the TDR  as shown in Fig. \ref{fig:eandb}. Therefore, the lower bound of $S^{D}_{a\leftarrow T}$ for given pairwise separation $d$ is:
\begin{equation}
S^{D}_{a\leftarrow T} \ge \pi \left(\frac{d}{2}\right)^{2}
\end{equation}
It is a monotone increasing function of $d$, which means that the larger is the pair-wise separation among the targets, the larger is the area of TDR for each target. 


%
\begin{figure}[htpb]
    \begin{center}
	\includegraphics[width=.20\textwidth]{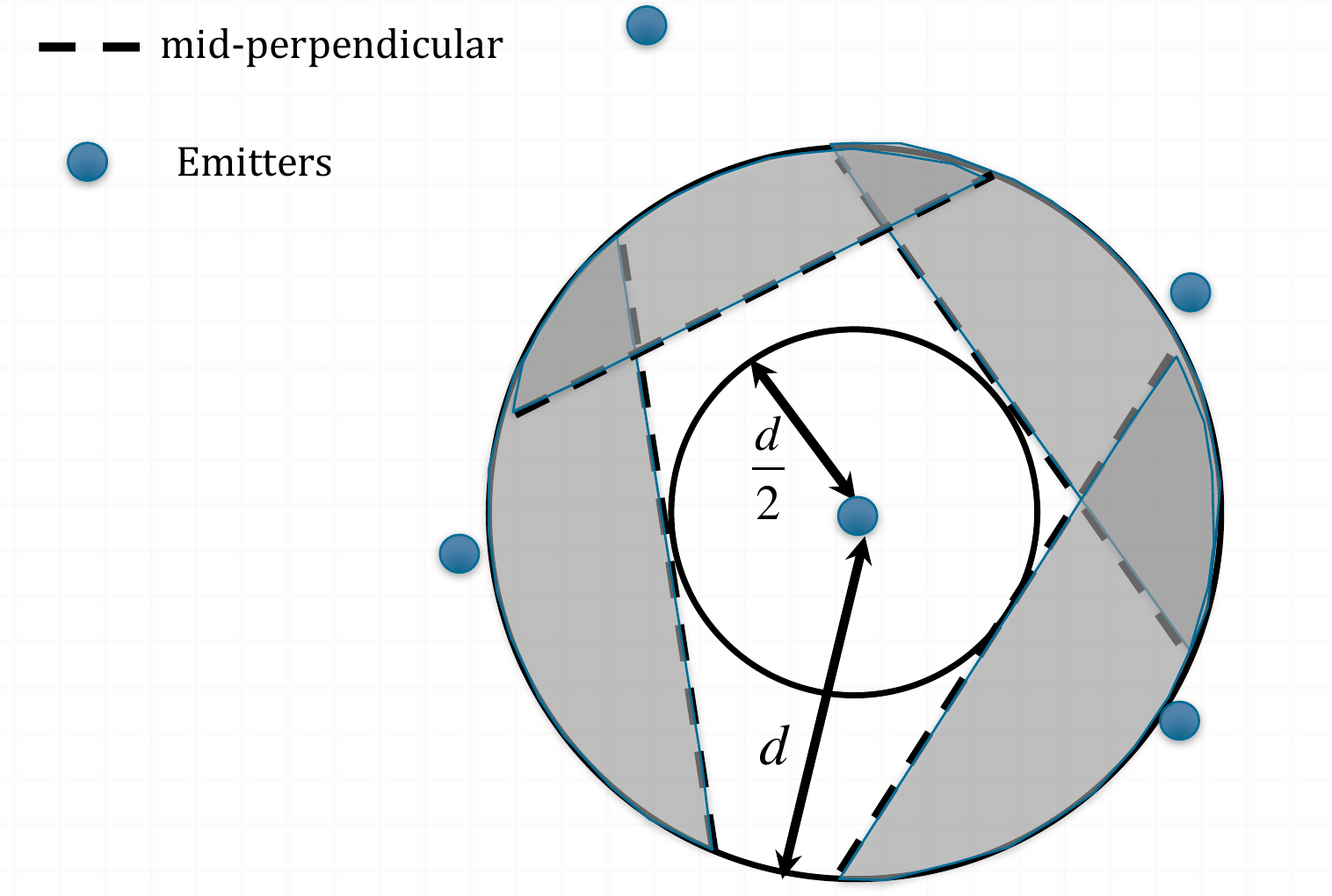}
    \end{center}
    \caption{Lower bound of blind-region}
    \label{fig:eandb}
\end{figure}

\subsection{Probability of Having At Least Three Receivers in TDR}
Based on the lower bound of $S^{D}_{a\leftarrow T}$, for any given distribution of the receivers, we can evaluate the probability and the expectation of at least three receivers in the TDR region of $a$. Note that different formulas can be utilized to estimate the lower bound of $S^{D}_{a\leftarrow T}$ if we know the number of concurrent targets in the audible region and the minimum separation distance $d$. 

Let's consider a general case when the receivers are in Poisson distribution, i.e.$P(n_{r}=k)=\frac{\lambda^{k}e^{-\lambda}}{k!}$, where $\lambda$ is the expected number of receivers in a unit area (e.g. 1 $m^{2}$). By substituting the lower bound of $S^{D}_{a\leftarrow T}\ge \pi \left(\frac{d}{2}\right)^{2}$,  the probability of at least three receivers are in $S^{D}_{a\leftarrow T}$ can be calculated as:
\begin{equation}
1 - \sum\limits_{i = 0}^2 {p({n_r} = i) \ge } 1 - {e^{ - \frac{{\lambda \pi {d^2}}}{2}}}\left[ {1 + \frac{{\lambda \pi {d^2}}}{2} + \frac{{{\lambda ^2}{\pi ^2}{d^4}}}{8}} \right]
\end{equation} 
\begin{theorem}
\emph{When receivers are in Poisson distribution with $\lambda$ expected receivers in a unit area, when the pair-wise separation among targets are larger than $d$, the probability of at least three receivers are presenting in the TDR of a target is lower bounded by
\begin{equation}
1 - {e^{ - \frac{{\lambda \pi {d^2}}}{2}}}\left[ {1 + \frac{{\lambda \pi {d^2}}}{2} + \frac{{{\lambda ^2}{\pi ^2}{d^4}}}{8}} \right].
\label{equ:prolb}
\end{equation}
}
\end{theorem}

Fig.\ref{fig:prob} plots the lower bound of $P(n_{r}\ge 3)$ as a function of $d$ and $\lambda$. We can see that for given $\lambda$, the lower bound of at least three receivers presenting in the TDR of a target increases exponentially with $d$. Note that the figure plots only the lower bound. Because the real TDR area can be much larger than the lower bound area of TDR, in real case, the probability of three receivers are in the TDR of a target can be much closer to 1. 
\begin{figure}[htpb]
    \begin{center}
	\includegraphics[width=.25\textwidth]{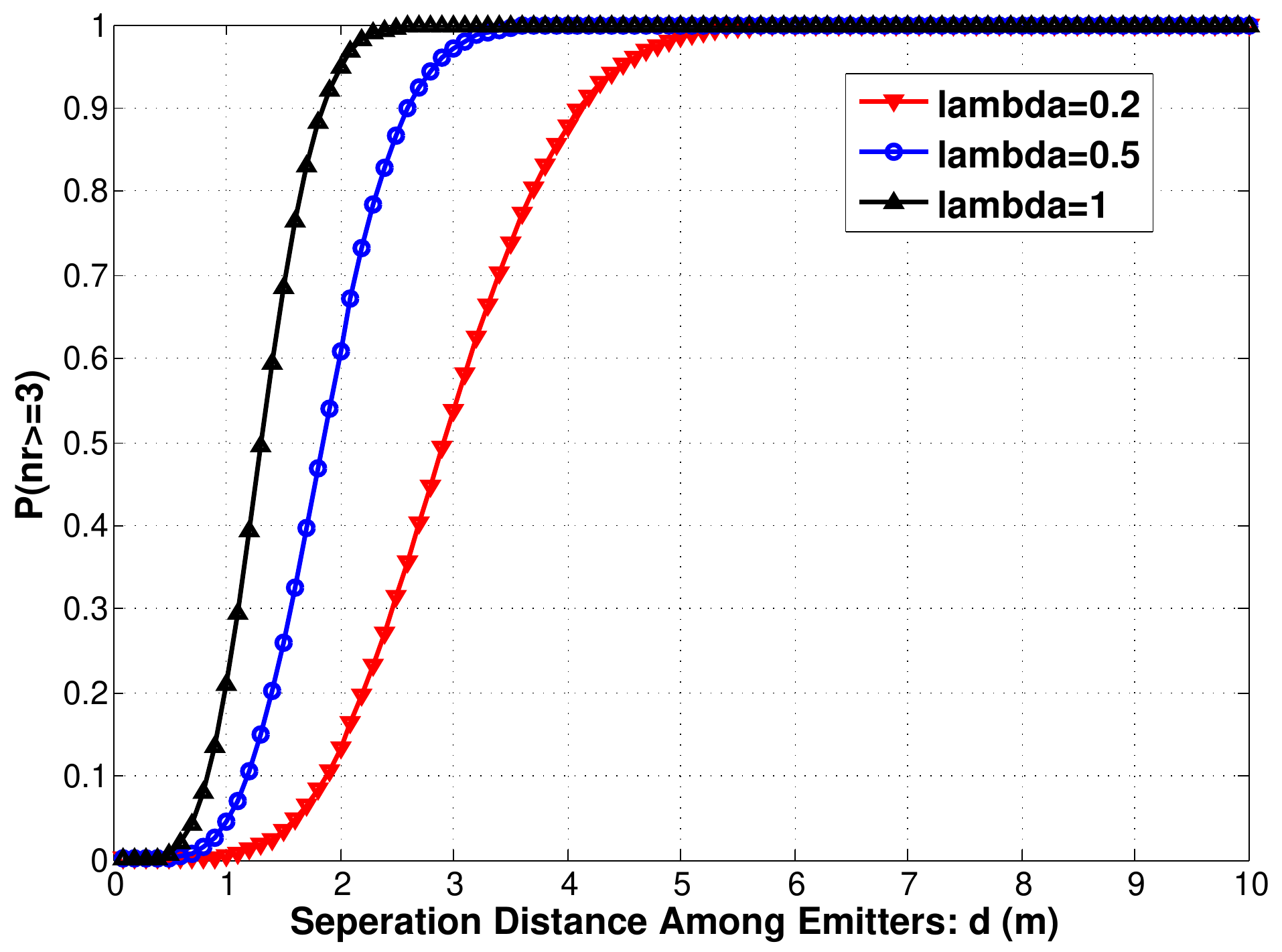}
    \end{center}
    \caption{The lower bound of the probability  of at least three receivers are in the TDR of a target as a function of $d$ and $\lambda$}
    \label{fig:prob}
\end{figure}
The results in Fig.\ref{fig:prob} show the strong feasibility of chorus locating. It only needs the targets are well separated and the receivers have enough density for receivers to obtain at least three TOA-based distances for each concurrent target. 
\begin{figure*}[htpb]
    \begin{center}
	\includegraphics[width=0.9\textwidth]{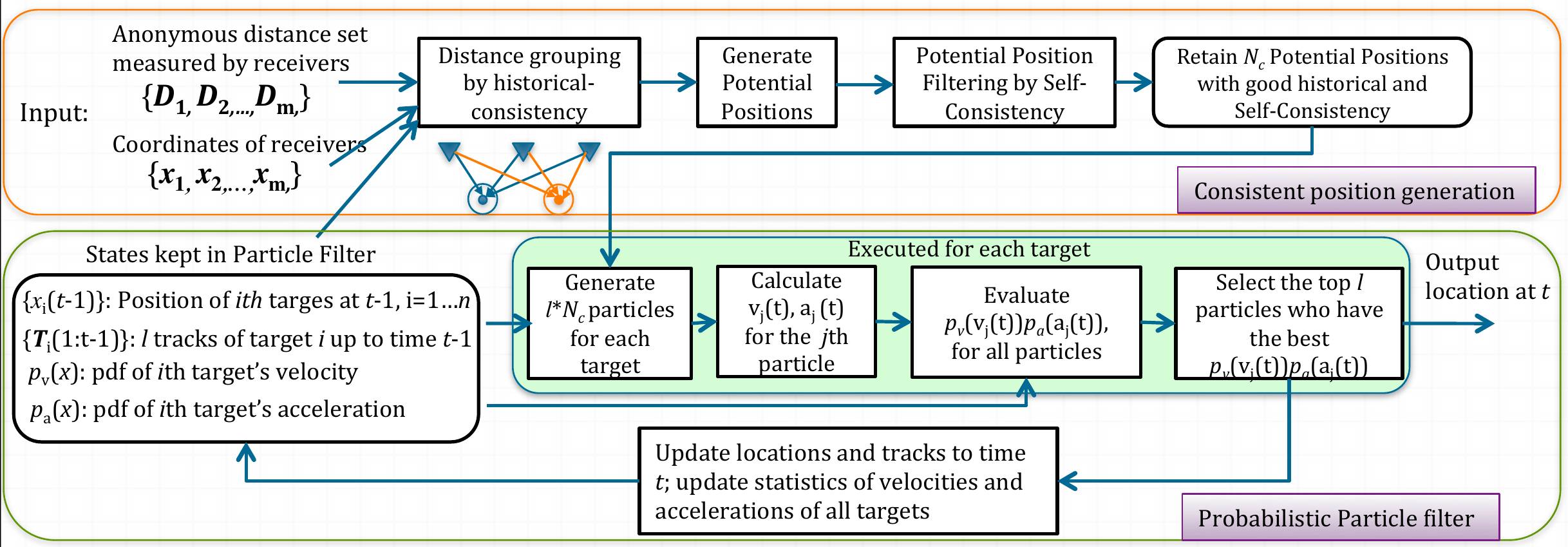}
    \end{center}
    \caption{The diagram of consistent location generation and probabilistic particle filter algorithms to utilize the anonymous distances measured by receivers to locate and disambiguate the tracks of multiple targets}
    \label{fig:algorithm}
\end{figure*}
\section{Locate Multiple Targets by Anonymous Distances}

Above analysis shows the feasibility of detecting multiple anonymous TOAs at a receiver in chorus mode. 
But the receiver don't know the source (target) of each TOA.  To utilize these TOAs to locate the multiple targets, we developed methods to effectively utilize the anonymous distances to locate the multiple targets and to disambiguate their trajectories. We introduce the proposed algorithms in this Section.  
\subsection{Overview}
The overview of the proposed techniques are shown in Fig.\ref{fig:algorithm}, which contain mainly two parts: 1) consistent position generation and 2) probabilistic particle filter for trajectory disaggregation. In the first part, the inputs are the set of anonymous distances measured by the receivers, denoted by $[\mathbf{D}_{1},\cdots,\mathbf{D}_{m}]$, and the coordinates of these receivers, denoted by $[\mathbf{x}_{1},\cdots,\mathbf{x}_{m}], $where $m$ is the number of receivers. The number of distances measured by the $i$th receiver is $\left| {{\mathbf{D}_i}} \right| = {k_i}$. 
\subsubsection{Overview of Consistent Position Generation}
Since each three distances from non-collinear receivers can generate a position estimation, enumerating the combinations of these anonymous distances will generate a large amount of possible positions, in which most of the positions are wrong. To avoid the pain of finding needless from the sea of large amount of potential positions, we proposed to firstly find the feasible distance groups by historical-consistency, i.e., by utilizing the consistency of distance measurements with the latest location estimations of the targets (which are provided by the particle filter). After this step, the distance groups are utilized to generate a set of potential positions. To further narrow down the potential position set, we proposed self-consistency to evaluate the residue of location calculation of each potential position. Only the top $N_{c}$ potential locations with good self-consistency will retained to be used as input to the particle filter at time $t$. 
\subsubsection{Overview of Probabilistic Particle Filter} 
The particle filter maintains the positions of $n$ targets at time $t-1$, denoted by $\{\mathbf{x}_{i}(t-1)\}$; maintains $l$ most possible tracks for each target up to time $t-1$, denoted by $\{\mathbf{T}_{i}(1:t-1)\in R^{l * (t-1)}\}$; the probability distribution function (pdf) of each target's velocity, denoted by $p_{v}(x)$; and the probability distribution function of each target's acceleration, denoted by $p_{a}(x)$. 
Then at time $t$, for each target $i$, by connecting its $l$ tracks at time $t-1$ to $N_{c}$ potential positions at time $t$,  $lN_{c}$ particles are generated. The velocity ($v_{j}(t), j=1, \cdots, l*Nc$) and acceleration ($a_{j}(t), j=1, \cdots, l*Nc$) of each particle are calculated, based on which,  the likelihood of the particle $j$ is evaluated by $p_{v}(v_{j}(t))p_{a}(a_{j}(t))$. Then by ranking the likelihoods of the particles, $l$ top particles will be retained for target $i$ at time $t$, which are used to update the location estimation of target $i$ at time $t$, the historical tracks and the pdfs of velocity and acceleration.  We introduce key points of the algorithm in following subsections. 
\subsection{Consistent Potential Position Generation }
\subsubsection{Historical Consistency}To avoid generating a large amount of misleading potential positions by blind combinations of the anonymous distances,  we proposed to measure the \emph{historical consistency} of the distances to label the distances to reasonable sources . The input of this step is the historical positions of the $n$ targets provided by particle filter and the distance set from the receivers.  For a target, since the velocity of the target is upper-bounded in the real scenarios, which is denoted by $v_e$, its position at time $t$ will be bounded inside a disk centered at its position at $t-1$, with radius $v_{e}$, i.e., 
\begin{equation}
    ||\mathbf{x}_{i}(t)-\mathbf{x}_{i}(t-1)||\le v_e 
    \label{equ:byi}
\end{equation}
For a receiver $j$, let $d_{{j}, i^{(t-1)}}$ represent the distance from it to $\mathbf{x}_{i}(t-1)$. From triangular inequality, for every distance $D_{k}$ measured by receiver $j$ at time $t$, $D_{k}$'s potential source is labeled to target $i$ if:
\begin{equation}
|D_{k}-d_{{j},i^{(t-1)}}|\le v_{e}
\end{equation}  
Then, only the distances with the same source (target) label will be selected to generate potential positions for the targets using trilateration. This step on one hand reduces the computation cost of generating massive possible positions, on the other hand avoids generating the obviously wrong positions. 
\subsubsection{Self-Consistency} We further evaluate the self-consistency of the generated potential positions to further filter out the unreasonable position candidates.  Considering a potential position $\mathbf{x}$ calculated  by trilateration using $m$ distances $[D_{1}, \cdots, D_{m}]$ from receivers at location $\mathbf{x}_{r_{1}}, \cdots, \mathbf{x}_{r_{m}}$, the self-consistency of this location is measured by the residue of  the location calculation:
\begin{equation}
S_{x}=\frac{1}{m}\sum_{i=1}^{m}\left(D_{i}-d_{x,r_{i}}\right)^{2}  
\end{equation}
where $d_{x,r_{i}}$ is the distance from $\mathbf{x}$ to receiver $\mathbf{x}_{r_{i}}$. Then only top $N_{c}$ potential positions with the best self-consistency performances will be retained as the input for particle filter to be further processed by particle filter at time $t$.

\subsection{Probabilistic Particle Filter}
The particle filter maintains 1) the locations of $n$ targets at $t-1$; 2) $l$ most possible tracks for each target up to time $t-1$, and 3) the probability density functions (pdfs) of each target's velocity and acceleration. The pdfs of each target's velocity and acceleration are calculated based on historically velocity and acceleration up to $t-1$.  They are utilized to evaluate the likelihood of the generated particles. 
\subsubsection{Generate and Evaluate Particles} 
For each target, say $i$, by connecting its $l$ ending locations at $t-1$ (in its $l$ tracks) to the $N_{c}$ potential positions at time $t$, $l*N_{c}$ particles are generated, each particle represents a potential track. Then we evaluate the likelihood of each particle $k, k=1,\cdots, l*N_{c}$ by the following likelihood function:
\begin{equation}
c_{k}=p_{v}(v_{k}(t))p_{a}(a_{k}(t))
\end{equation}
where $v_{k}(t)$ and $a_{k}(t)$ are calculated on the particle $k$ by:
\begin{equation}
	v_{k}(t)=|\mathbf{x}_{k}(t)-\mathbf{x}_{k}(t-1)|, a_{k}(t)=v_{k}(t)-v_{k}(t-1)
    \label{equ:VandA}
\end{equation}
Then the top $l$ particles with best likelihood will be retained for the target for the next step, and $\mathbf{x}(t)$ in the most possible particle will be output as the position estimation at time $t$. The pdfs of velocity and acceleration are updated accordingly. Such a progress will be applied to all the targets, and the algorithm of the probabilistic particle filter is listed in Algorithm \ref{algo:UpdateParticle}.
\begin{figure}[t]
    \begin{center}
	\includegraphics[width=.4\textwidth]{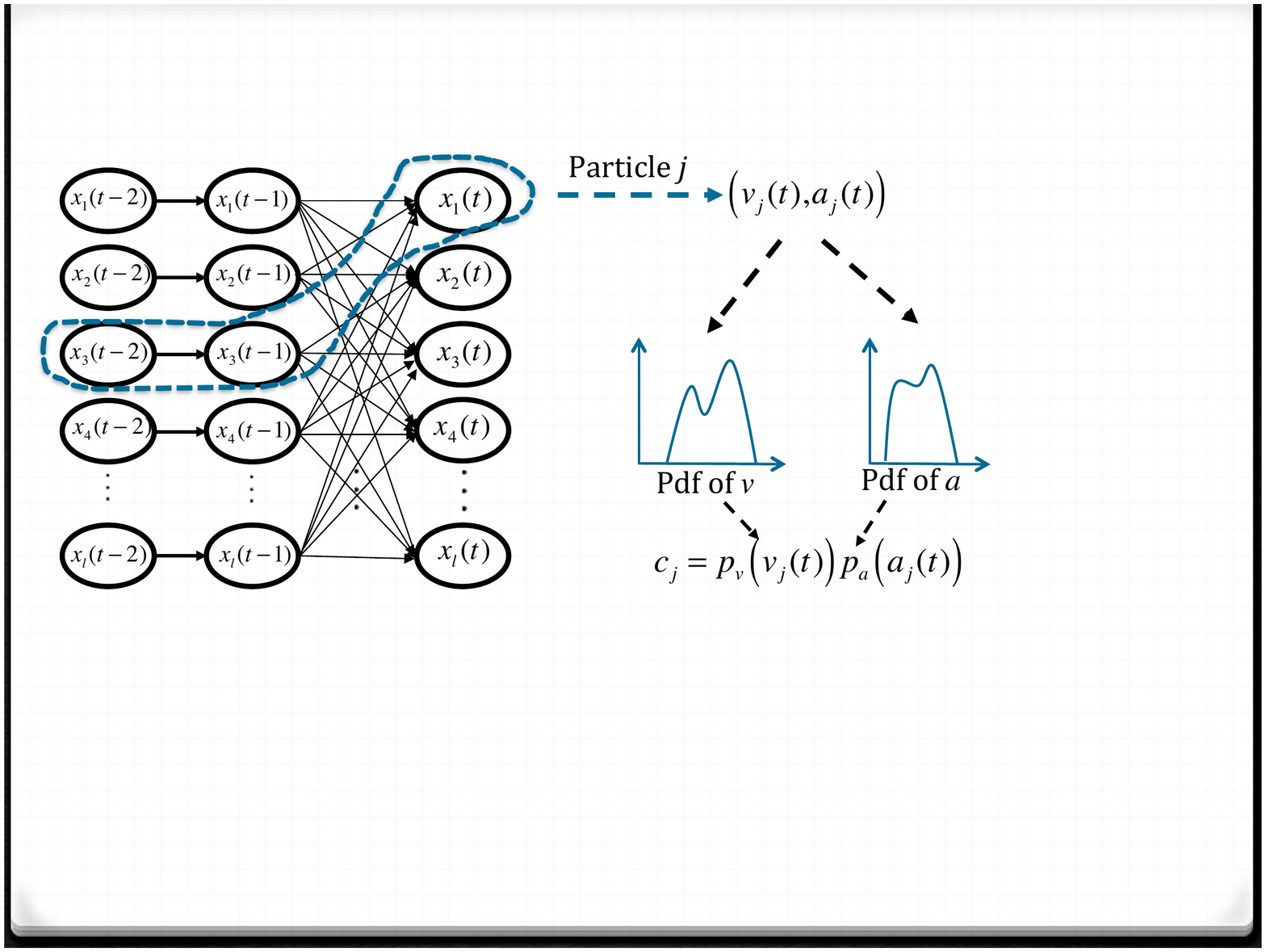}
    \end{center}
    \caption{Evaluate the cost of each generated particle}
    \label{fig:ParticleFilter}
\end{figure}

\begin{algorithm}[b]
    \caption{Probability Particle Filter for a Target $i$}
    \begin{algorithmic}[1]
	\REQUIRE $\mathbf{T}_{i}(\!1\!:\!t-1\!)$ , possible 
	location $\{\mathbf{x}_1,$ $\mathbf{x}_2,\dots,\mathbf{x}_{n_c}\}$. PDF of velocity
	$p_v(\cdot)$ and PDF of acceleration $p_a(\cdot)$.  
	\ENSURE Updated $\mathbf{T}_{i}(1:t)$, $p_v(\cdot)$ and $p_a(\cdot)$, $\mathbf{x}_{i}(t).$ 
	\STATE $\{p_1,\dots,p_{l\times n_c}\}\leftarrow\mathbf{T}_{i}(\!1\!:\!t-1\!)\times \{\mathbf{x}_1,\dots,\mathbf{x}_{n_c}\}$ // Generate particles by posible locations of tracks at $t-1$ 
	\STATE $\{c_1,\dots,c_{l\times n_c}\}\leftarrow \mathbf{0}$
	\FOR {$i=1 : l\times n_c$ }
	\STATE $v_{k}(t)=|\mathbf{x}_{k}(t)-\mathbf{x}_{k}(t-1)|$
	\STATE $ a_{k}(t)=v_{k}(t)-v_{k}(t-1)$
	\STATE $c_k=p_v(v_{k}(t))\cdot p_a(a_{k}(t))$
	\ENDFOR
	\STATE $\{\hat{p}_1,\dots,\hat{p}_{l\times n_c}\} \leftarrow$ sorting $\{p_1,\dots,p_{l\times n_c}\}$ by $\{c_1,\dots,c_{l\times n_c}\}$ in ascending order
	\STATE $\mathbf{T}_{i}(1:t)\leftarrow \{\hat{p}_1,\dots,\hat{p}_{l}\}$  //
	preserve the first $l$ sorted particle  
	\STATE $p_a(\cdot)\leftarrow UpdatePDF(p_a(\cdot),\{v_1(t),\dots,v_l(t)\})$
	\STATE $p_v(\cdot)\leftarrow UpdatePDF(p_v(\cdot),\{a_1(t),\dots,a_l(t)\})$
	\STATE $\mathbf{x}_{i}(t)=\hat{p}_1$
    \end{algorithmic}
    \label{algo:UpdateParticle}
\end{algorithm}

Complexity of Algorithm\ref{algo:UpdateParticle} can be easily verified.
\begin{lemma}
    Complexity of algorithm \ref{algo:UpdateParticle} is $O(n N_cl\log(N_cl) )$
    \label{lemma:pt2}
\end{lemma}
\begin{proof}
For each target, the most expensive step is to sort the $l*N_{c}$ elements, which takes $O(N_cl\log(N_cl)))$, so the overall complexity for locating the $n$ targets is $O(n N_cl\log(N_cl)))$.
\end{proof}

The probabilistic particle filter provides good flexibility. 1) It supports the trade off between the locating accuracy and the executing time by changing the number of the preserved particles. 2) The likelihood of each Particle is calculated by considering both the velocity and the acceleration, which is online continuously updated, so that it can be suitable even when the targets have variated motion characters. 

A potential drawback of this particle filter approach is that a target may be lost when it is too close to other targets. When the  location candidates of two targets are almost the same, all particles may follow one target and none particle
follows the other. Although such kind of target lost happens only in chance, it affects the tracking performance occasionally. We show this problem can be well solved by location based time-slot scheduling, which is discussed in the next section.  
\section{Location based Time-slot Scheduling}
Locating in chorus mode requires concurrent targets have enough pair-wise separation distances, otherwise the receivers cannot detect TOAs from their concurrent waves. Keeping concurrent targets to be spatially well separated is also important for the particle filter to confidentially disambiguate their tracks. In addition, the initial condition of the particle filter needs the initial location estimations to be as accurate as possible to avoid cascading errors. 

With consideration of these requirements, we designed location based time-slot assignment (LBTA) to appropriately schedule the concurrent transmissions of the targets. In general,  LBTA assigns targets which are close to others or with unknown locations to work in exclusive time-slots to avoid conflict. Targets satisfying the separation distance are scheduled to transmit concurrently. 

At first in LBTA, a confident separation distance $d_{s}$ is calculated  by the lower bound of TDR region (\ref{equ:prolb}) based on given density of the receivers, i.e., $\lambda$ to guarantee $P(n_{r}\ge 3)$ approaching 1. Then the targets with known locations will be separated into a set of \emph{$d_{s}$-separated groups}. Each group consists of several targets with the pair-wise distance among targets in the group is at least $d_{s}$.  Then an exclusive time slot is assigned to the targets in the same \emph{$d_{s}$-separated group}. Exclusive slots are also assigned to the targets with unknown locations.  

\begin{figure}[h]
    \begin{center}
	\includegraphics[width=.3\textwidth]{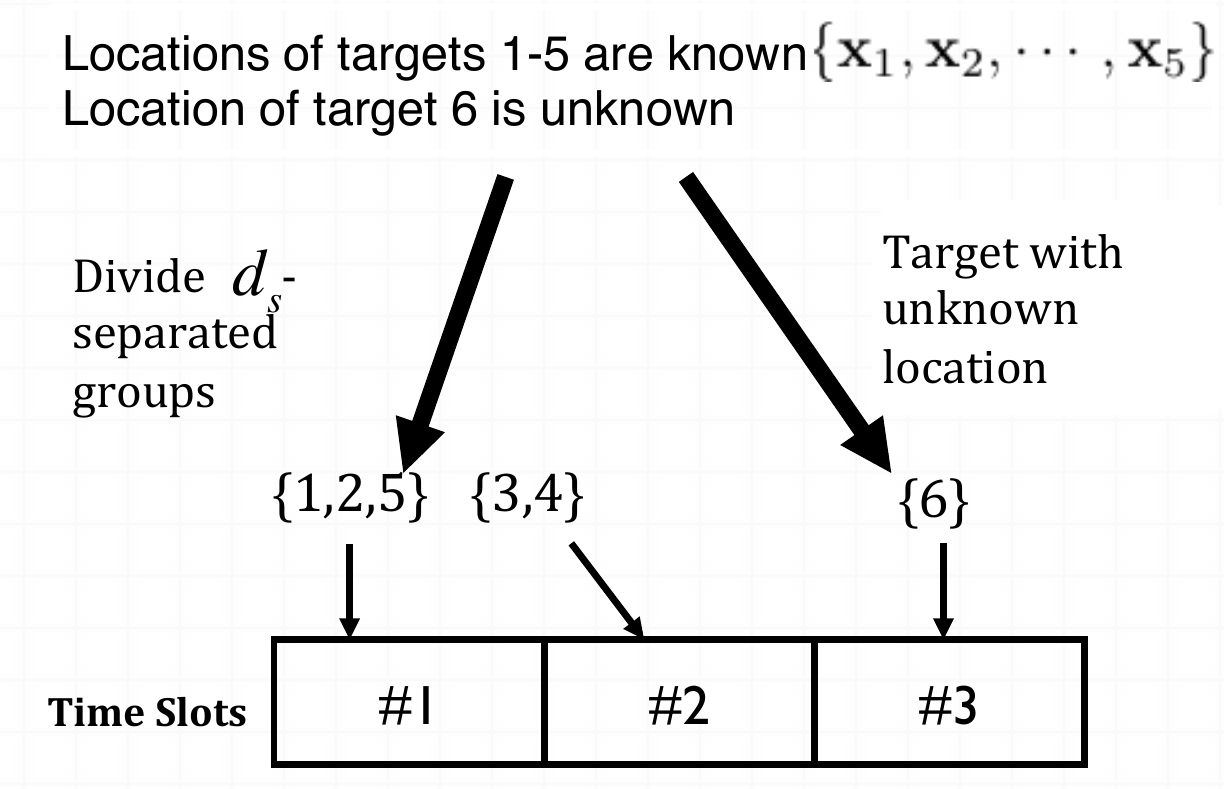}
    \end{center}
    \caption{An example of LBTA to assign time slots}
    \label{fig:LBTA}
\end{figure}
An example of LBTA is shown in Fig. \ref{fig:LBTA}, in which, six targets are presenting. We assume the locations of target $\{1,\dots,5\}$ are known and the locations of targets $6$ is still unknown. In this case, the targets with known locations are separated into two \emph{$d_s$-seperated groups}.  An exclusive time-slot is assigned to every $d_s$-seperated group and the target with unknown location.

\emph{LBTA} can help to solve both initialization problem and the risk of missing target in particle filter. At initial state, location of all $n$ targets are unknown. So $n$ time slots are required to locate the $n$ targets. From then on,  all $n$ targets share one time slot unless pairwise distances between some targets are less than $d_s$. In this case, partition method on the $n$ target is used to separate the targets into $d_s$-separated groups. Although finding the minimum number of  $d_s$-seperated group is NP-hard\cite{Fishkin:2004ty}, this problem can be effectively addressed by a greedy approach in practice when the number of targets are limited. We proposed a greedy \emph{DivideClosestTargets} algorithm to address it. 
\begin{algorithm}[t]
    \caption{DivideClosestTargets}
    \begin{algorithmic}[1]
	\REQUIRE $\{\mathbf{x}_1,\dots,\mathbf{x}_{n}\}$ and $d_s$
	\ENSURE $d_s$-seperated group partition, $\mathcal{G}_1,\dots,\mathcal{G}_{n_d}$  
	\STATE $n_d\leftarrow 1$, $\rm{tempg}_1\leftarrow\{\mathbf{x}_1,\dots,\mathbf{x}_{n}\}$, $\rm{tempg}_2=\emptyset$
	\WHILE {$\cup_{i=1}^{n_{d}-1} \mathcal{G}_i  \ne \{\mathbf{x}_1,\dots,\mathbf{x}_{n}\}$}
		\WHILE {$(\rm{MinPairWiseDis(\rm{temp}\mathcal{G}_1)}<d_{s})$}
		        \STATE $[i,j]=$ select the closest pair in $\rm{tempg}_1$
			\STATE $\rm{temp}\mathcal{G}_1=\rm{temp}\mathcal{G}_1\setminus i$, $\rm{temp}\mathcal{G}_2 =\rm{temp}\mathcal{G}_2 + i$
		\ENDWHILE
		\STATE $\mathcal{G}_{n_{d}} = \rm{temp}\mathcal{G}_1$, $n_{d}=n_{d}+1$
		\STATE $\rm{temp}\mathcal{G}_1 =\rm{temp}\mathcal{G}_2 $, $\rm{temp}\mathcal{G}_2 =\emptyset$
	\ENDWHILE
    \end{algorithmic}
    \label{algo:RandSplit}
\end{algorithm}
The algorithm always selects the closest pair in the current temp group, and put one of them into a new temp group, until all targets in current temp group have pairwise distance larger than $d_{s}$. This temp group will form a $d_{s}$-separated group. Then the algorithm process the new temp group, until all targets are assigned into $d_{s}$-separated groups. 
%
\section{Evaluation}
Both simulations and experiments were conducted to evaluate the  performances of multiple target locating in chorus mode. More specifically, the locating accuracy, efficiency of scheduling and,  robustness of chorus locating against noise were evaluated and reported in this section. 
\subsection{Simulation}

\begin{figure}[htpb]
    \begin{center}
	\includegraphics[width=.25\textwidth]{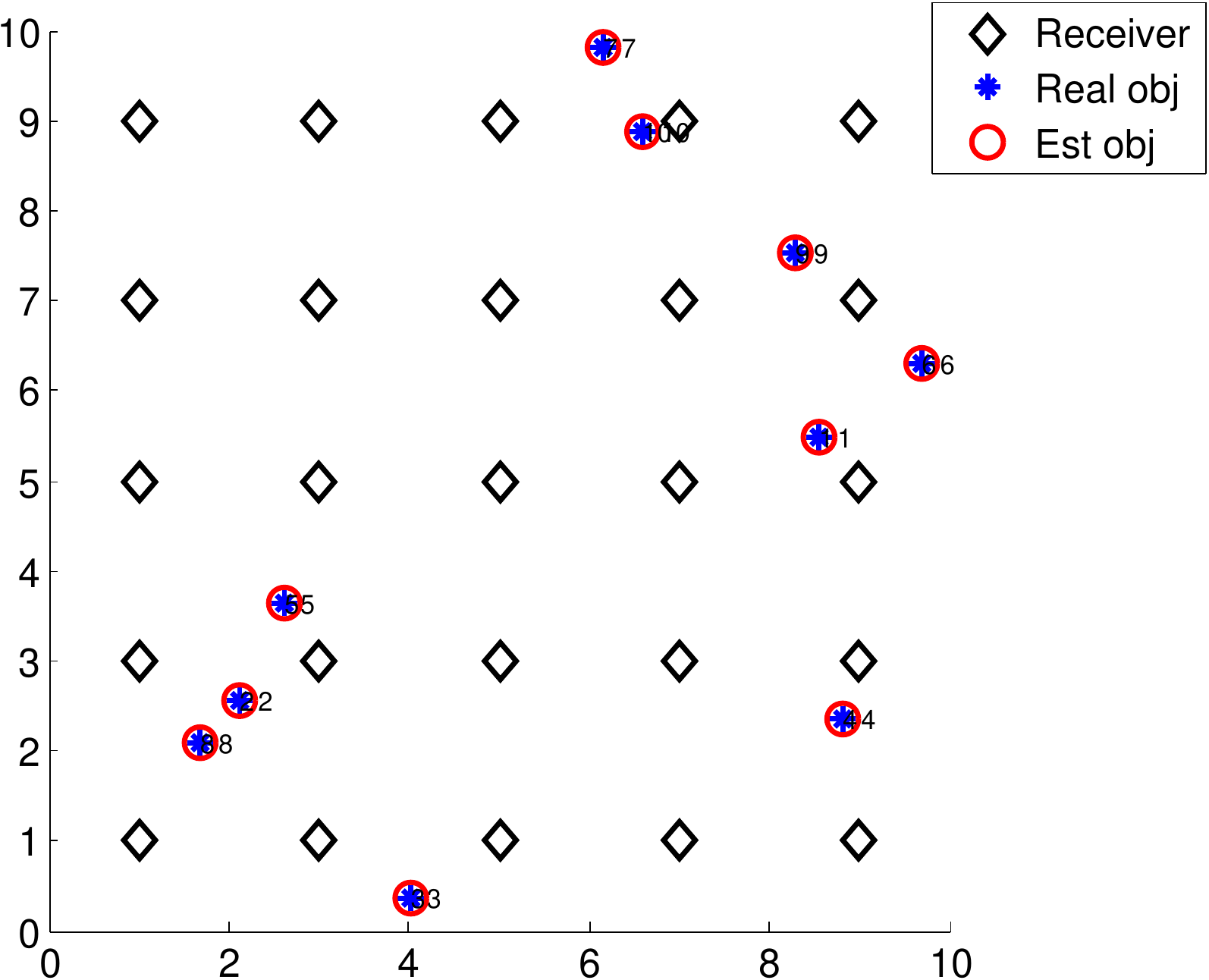}
    \end{center}
    \caption{Settings of simulation for chorus locating.}
    \label{fig:Setting}
\end{figure}
\subsubsection{Settings of Simulation}We conducted simulation by developing a multi-agent simulator in MATLAB environment. The setting of our simulation scenario is shown 
in figure \ref{fig:Setting}. The black diamonds stand for receivers, which are deployed in grid of
size $2m \times 2m$. The blue stars stand for targets. Motions of targets are identically
independent random walk, that each target walks along a line and turns a random angle every 5
seconds. The velocities of the targets are normally distributed, with $\mu=1$ and $\sigma=0.1$.
This motion character is close to real action of human in open space.  In simulation, we set the
number of targets to  10, whose actions are constrained in a box of size $10m\times 10m$. 
The length of a time slot, i.e., locating updating interval is set to $100ms$. The audible radius, i.e., $r$ of target is set to be $3m$. $\omega$, which the length of the aftershock is set to $0.33m$. The values of $r$, $\delta$ and $\omega$ in above setting are obtained from real values of Cricket \cite{Priyantha:2000hx} locating system.

\subsubsection{Locating accuracy without ranging noise} We firstly evaluate the multiple target locating and trajectory disaggregation performances when no ranging noise is incurred, i.e., ranging error is zero. The  accuracy for concurrently multiple target tracking is shown in figure \ref{fig:RandE} and \ref{fig:cdf}.  Fig.\ref{fig:RandE} plots the real trajectories and estimated trajectories, which shows that the estimated trajectories coincide well with the real trajectories even trajectories overlap. The corresponding CDF of the locating error is shown in figure \ref{fig:cdf}, which shows that more than $90\%$ of the locating error
is less than $1cm$. We found that greater than $1cm$ location error appeared when ranges was lost due to aftershock at a receiver resulting at $<3$ TOAs which leads to incorrect location estimation. 

\subsubsection{Accuracy vs. $\omega$ vs. time-slots}Location accuracy under different $\omega$ is shown in figure \ref{fig:CDFvsOmega}. The CDFs of ranging 
errors when $\omega$ equals to $33cm, 165cm, 330cm$ are presented, which are the corresponding cases when the length of the aftershock are $1ms, 5ms, 10ms$ respectively. Although the accuracy gets worse with growing of $\omega$, $90\%$ of the locating errors in the 3 cases are still very small. We investigated and found that the good locating performances against the variation of $\omega$ were contributed by LBTA. With the growth of the aftershock,  LBTA started to assign more time slots to the targets.  The slot assignment results are also shown in Figure \ref{fig:efficiency}, where the average number of concurrent targets located per times-slot are highly dependent on $\omega$.  With growing of $\omega$, the number of concurrently located targets per slot drops from 8 to 1.7. In other word, the chorus mode degenerated to the exclusive mode when $\omega$ is large, i.e., when the aftershock is long.

\subsubsection{Accuracy vs. ranging noises}Ranging noises are inevitable in ultrasound based locating systems, therefore noise resistance ability of chorus locating was also evaluated. To simulate the effect of ranging noise, positive offset is randomly added to every distance measurement. Offset is distributed from 0 to $l_o$ uniformly. The CDFs of locating errors with different $l_o (cm)$ is presented in Fig. \ref{fig:CDFvsNoise}, with $l_o$ being $1cm$, $5cm$ and $10cm$ respectively. The corresponding $90\%$-error is $1cm$, $10cm$ and $15cm$. Although there are no explicate anti-noise modules, it is shown that chorus locating can work under the impacts of the ranging noises. 
\begin{figure*}[htpb]
    \subfigure[Real and estimated trace]{
	\label{fig:RandE} 
	\begin{minipage}[b]{0.18\textwidth}
	    \centering
	    \includegraphics[width=\textwidth]{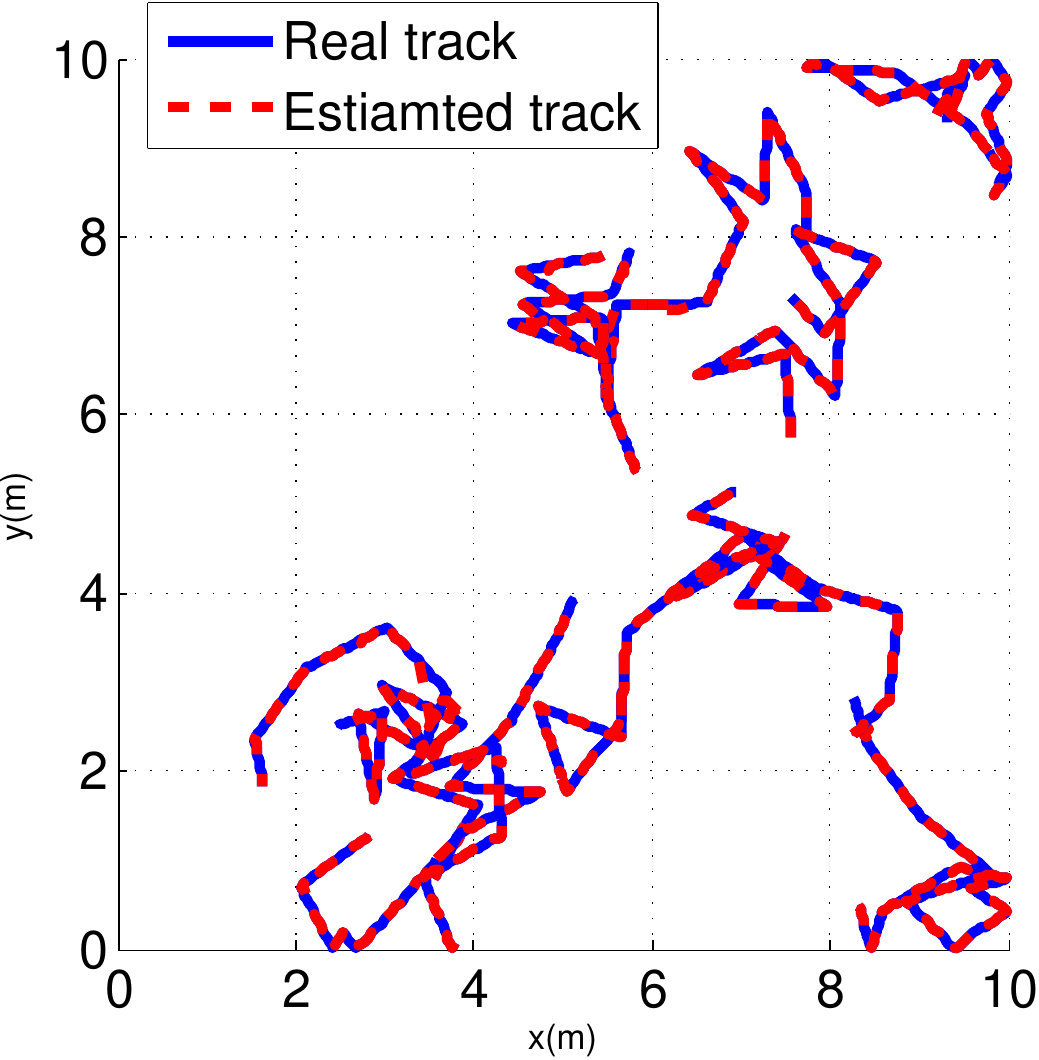}
	\end{minipage}}
    \subfigure[CDF of locating]{
	\label{fig:cdf} 
	\begin{minipage}[b]{0.18\textwidth}
	    \centering
	    \includegraphics[width=\textwidth]{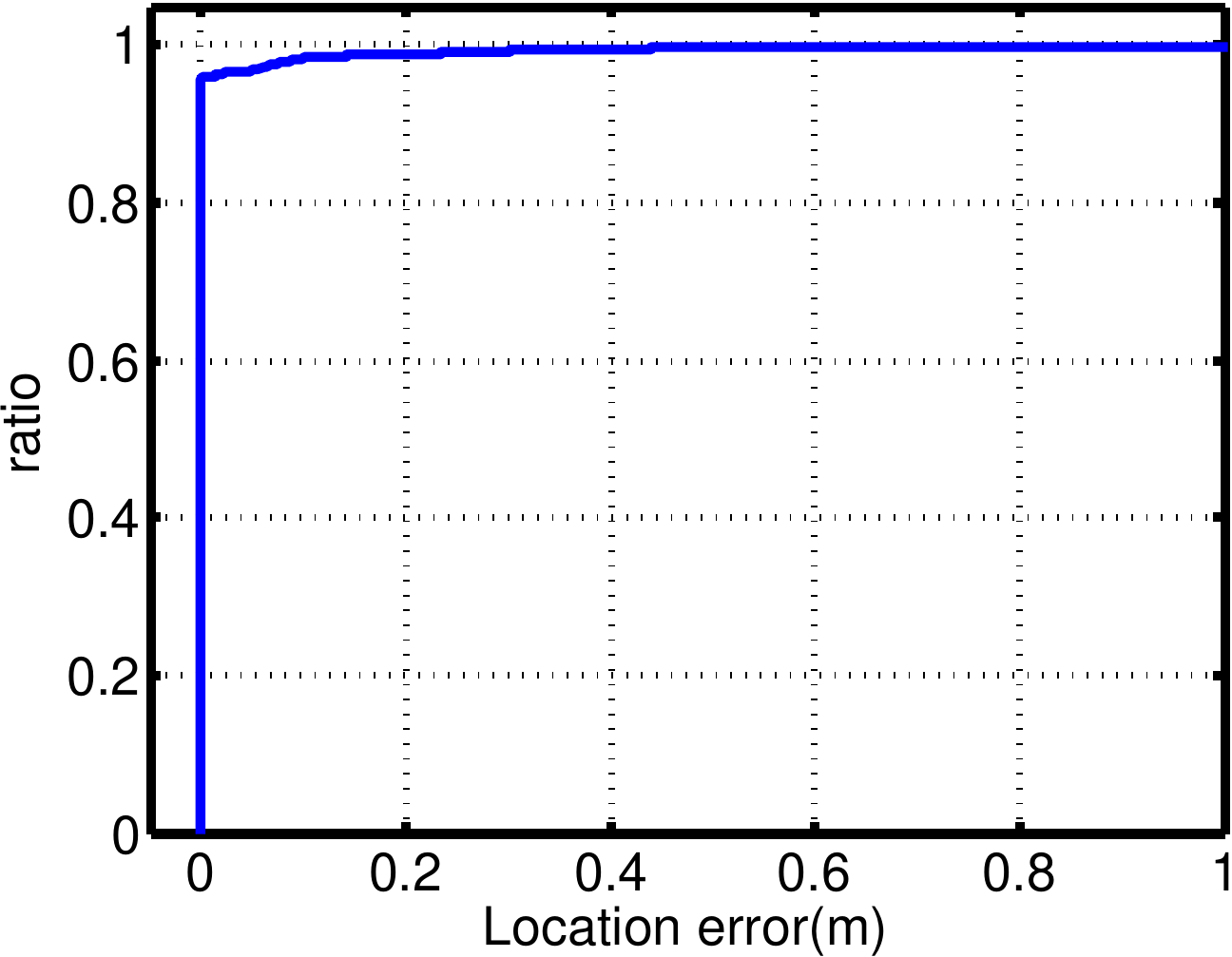}
	\end{minipage}}
    \subfigure[CDF vs. $\omega$ ]{
	\label{fig:CDFvsOmega} 
	\begin{minipage}[b]{0.18\textwidth}
	    \centering
	    \includegraphics[width=\textwidth]{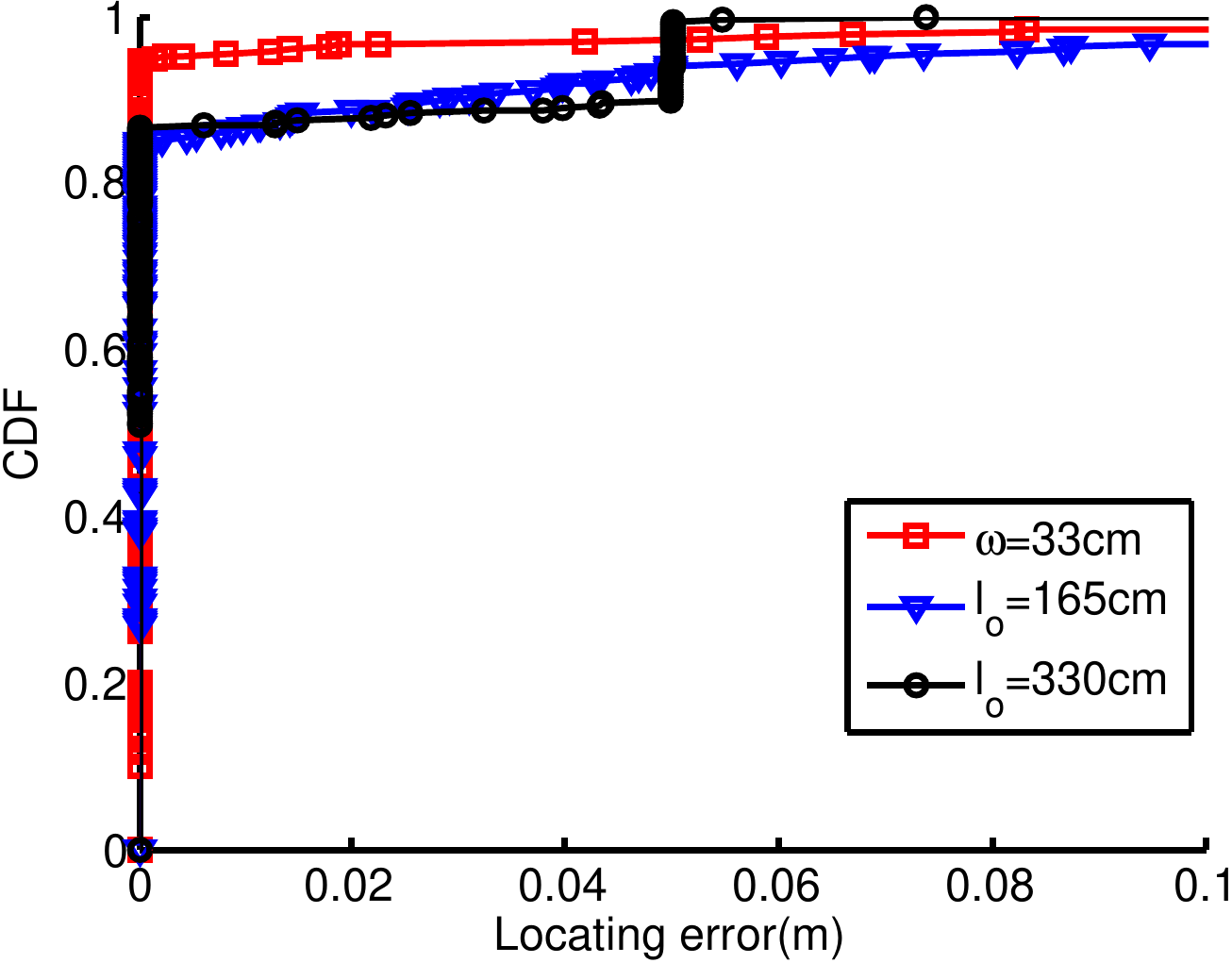}
	\end{minipage}}
    \subfigure[Efficiency vs. $\omega$ ]{
	\label{fig:efficiency} 
	\begin{minipage}[b]{0.18\textwidth}
	    \centering
	    \includegraphics[angle=90,width=\textwidth]{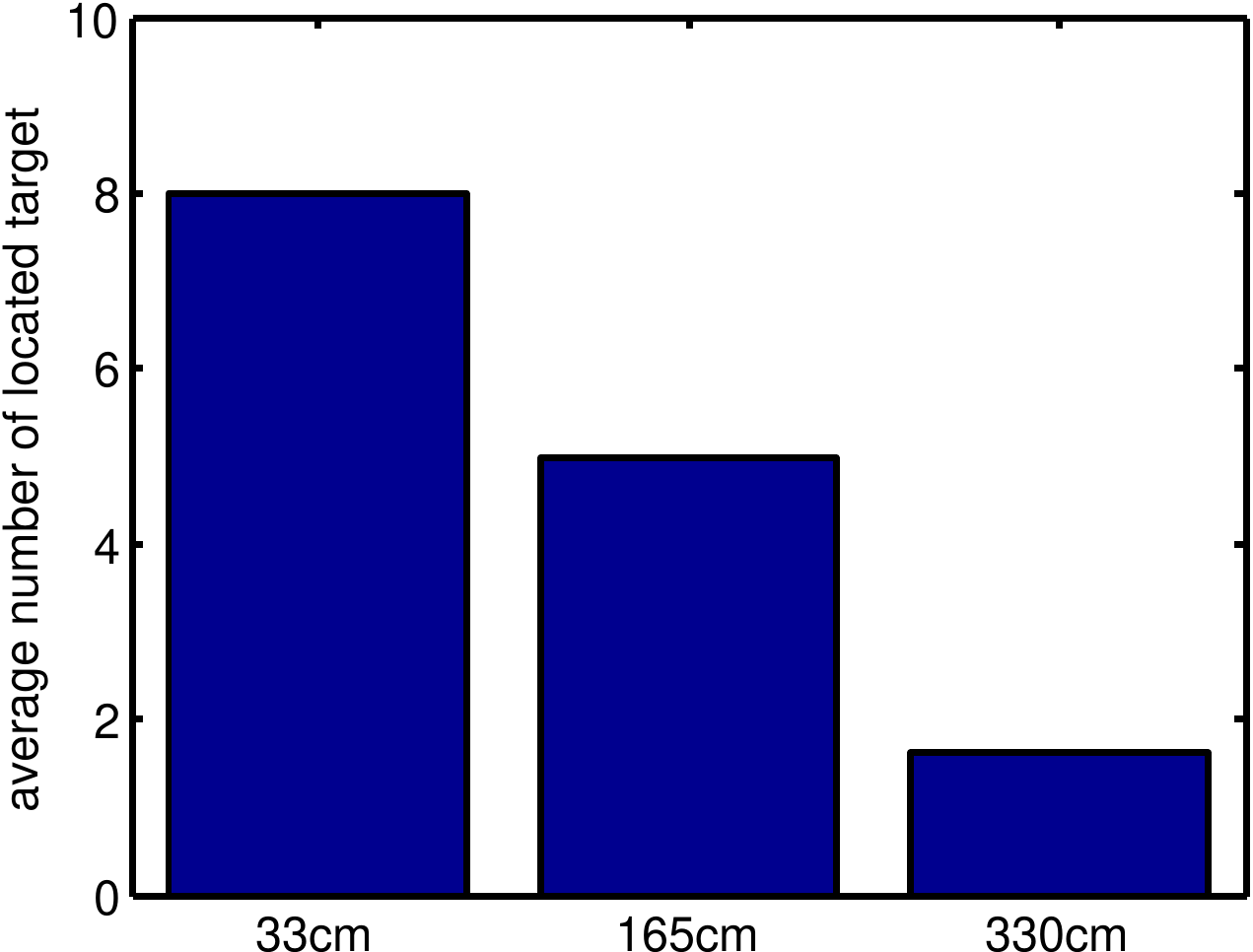}
	\end{minipage}}
    \subfigure[CDF vs. Noise ]{
	\label{fig:CDFvsNoise} 
	\begin{minipage}[b]{0.18\textwidth}
	    \centering
	    \includegraphics[width=\textwidth]{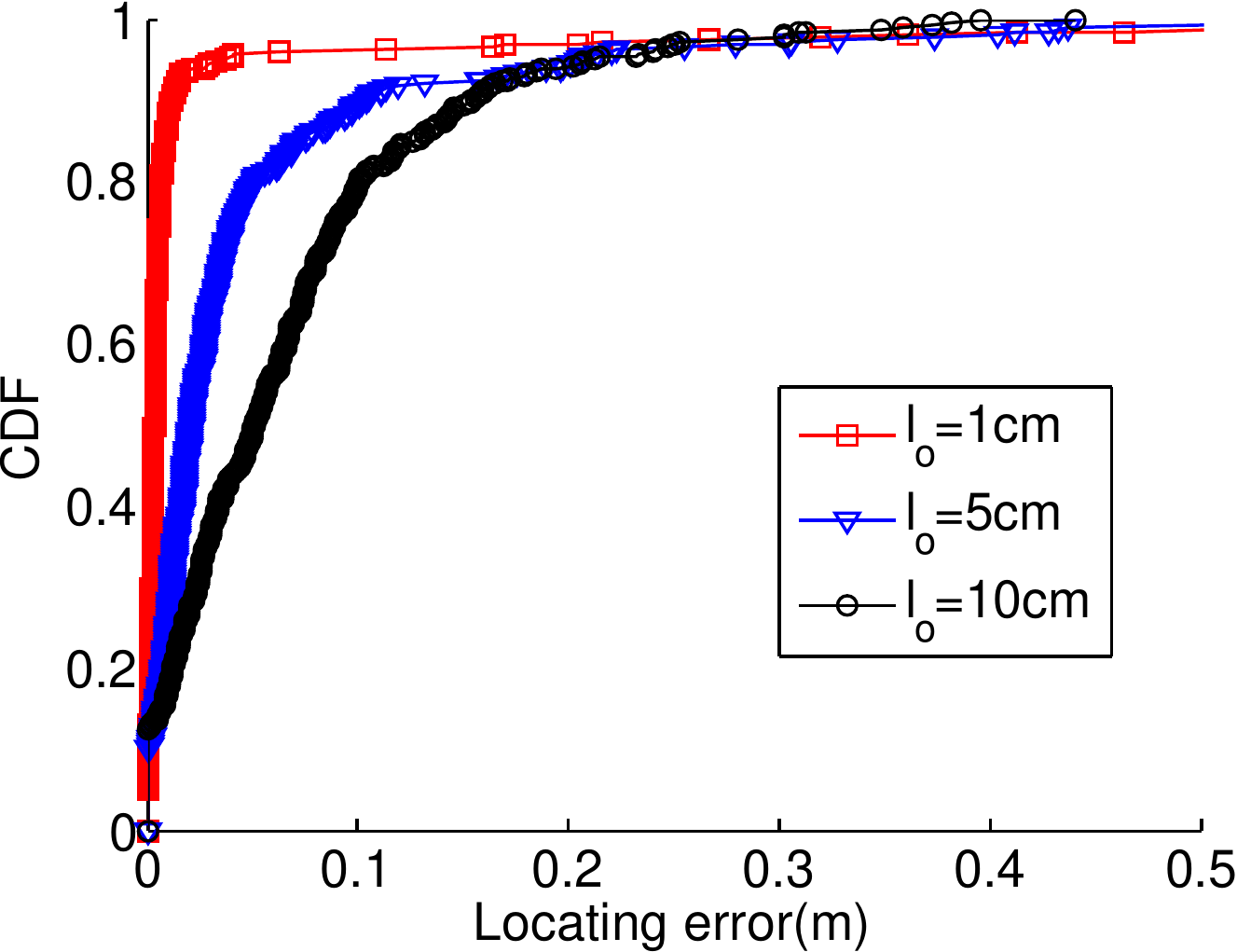}
	\end{minipage}}
    \caption{Performance evaluation obtained by simulation}
    \label{fig:PerformanceSimulation} 
\end{figure*}

\subsection{Testbed experiment}
We also conduct hardware experiments by using Cricket nodes. 4 nodes were tuned as receivers, which were deployed in an umbrella-type topology. Three nodes were programmed as targets, which were controlled by a sync-node. More specially, every target sends a NBU pulse once it hears the synchronizing signal from the sync-node. The time slot was set to $100ms$.  We modified the firmware of cricket, so that each receiver reports all detectable range measurements to a PC via rs232 cable.  Chorus locating algorithm was run at the PC end to calculate the locations for the multiple targets. The setting of the test-bed is shown in Fig.\ref{fig:testbed}. 

\begin{figure}[htpb]
    \subfigure[Receivers]{
	\label{fig:Receivers} 
	\begin{minipage}[b]{0.20\textwidth}
	    \centering
	    \includegraphics[width=\textwidth]{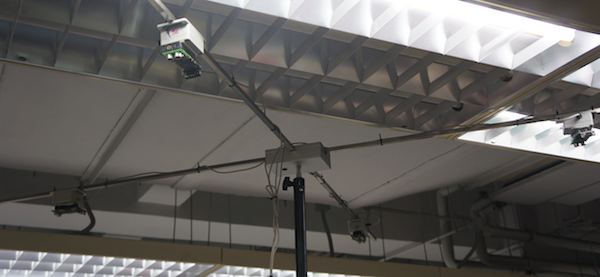}
	\end{minipage}}
    \hspace{.02\textwidth}
    \subfigure[Targets]{
	\label{fig:emitter} 
	\begin{minipage}[b]{0.25\textwidth}
	    \centering
	    \includegraphics[width=\textwidth]{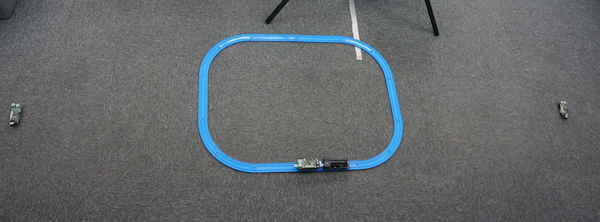}
	\end{minipage}}
    \caption{Setting of test-bed}
    \label{fig:testbed}
\end{figure}
 
Fig.\ref{fig:ExpTrace} shows the locating accuracy when a target $A$ was attached to a toy train,
which ran along a trail at $1m/s$, while two concurrent targets $b$ and $c$ were placed on the
ground. The locations of these concurrent targets were tracked by the four receivers. The obtained
trajectories of the target on the train  are presented in figure \ref{fig:ExpTrace}. Since it is
difficult to obtain the ground-truth of mobile target. CDF of static targets is presented in Fig.
\ref{fig:ExpCDF}. It is shown that more than $90\%$ of the locating errors is less than $15cm$. 

Therefore these simulation and experiment results verified the efficiency of locating multiple targets in chorus mode and the effectiveness of the proposed algorithms. They show that satisfactory accuracy can generally be obtained by locating in chorus mode. 
\begin{figure}[htpb]
    \subfigure[Trace of two target]{
	\label{fig:ExpTrace} 
	\begin{minipage}[b]{0.21\textwidth}
	    \centering
	    \includegraphics[width=\textwidth]{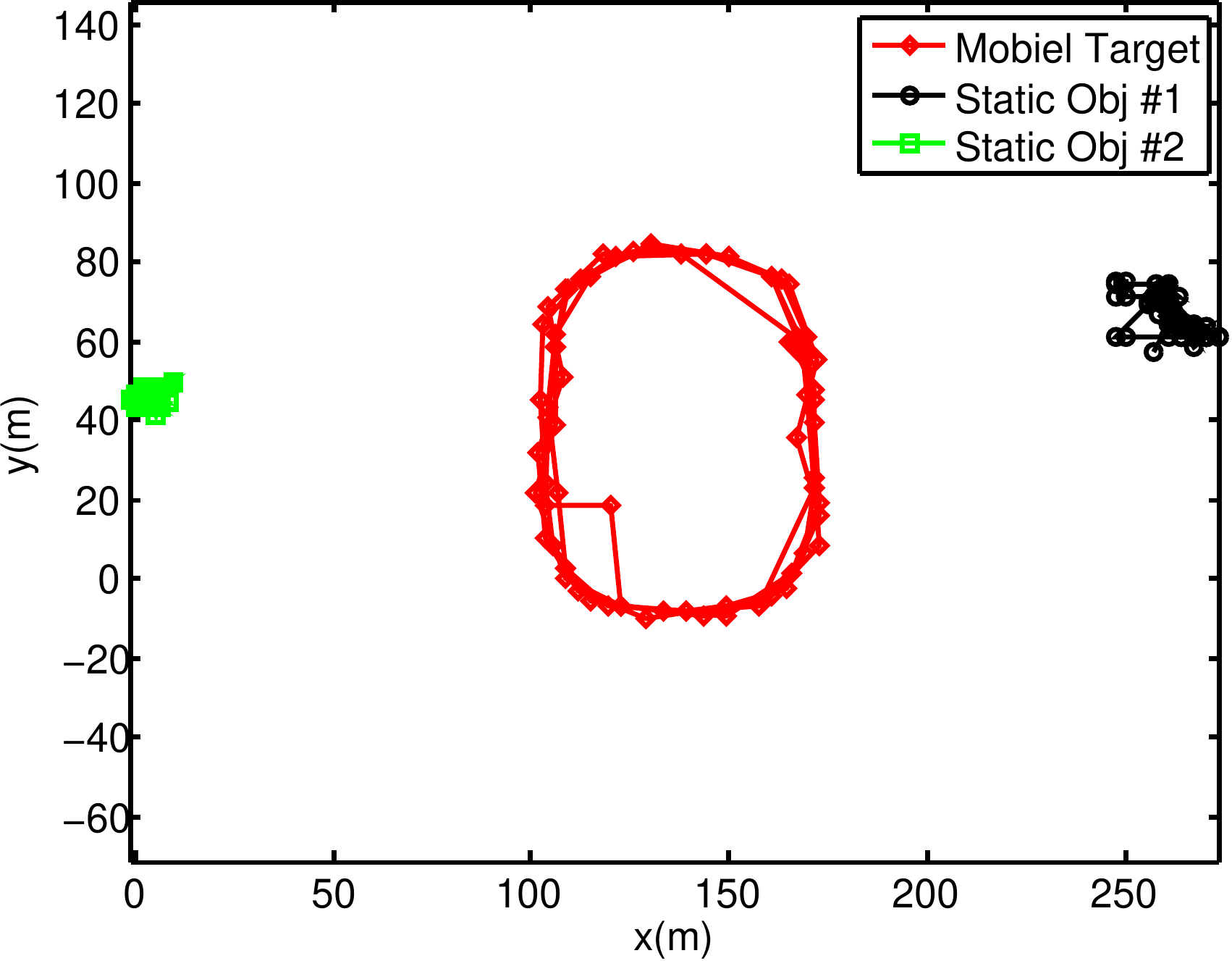}
	\end{minipage}}
    \hspace{.02\textwidth}
    \subfigure[Locating CDF of static target]{
	\label{fig:ExpCDF} 
	\begin{minipage}[b]{0.21\textwidth}
	    \centering
	    \includegraphics[width=\textwidth]{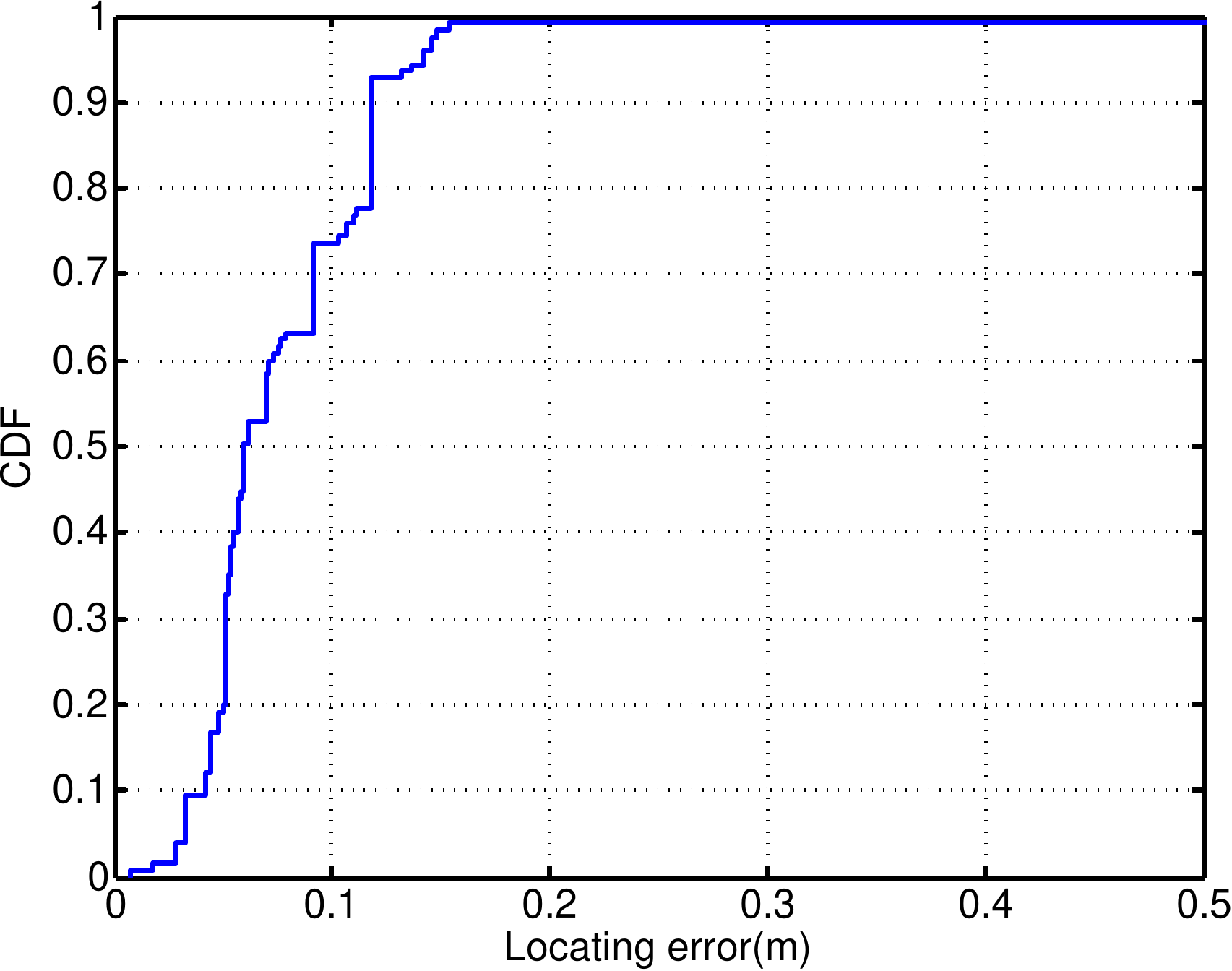}
	\end{minipage}}
    \caption{Performance evaluation obtained by testbed experiment}
    \label{fig:PerformanceExp} 
\end{figure}
\section{Conclusion}
We have investigated to locate multiple narrowband ultrasound targets in chorus mode, which is to allow the targets broadcast ultrasound concurrently to improve position updating rate, while disambiguating their locations by algorithms at the receiver end.  We investigated the geometric conditions among the targets for confidently separating the NBU waves at the receivers, and the geometrical conditions for obtaining at least three distances for each concurrent target. To deal with the anonymous distance measurements, we present consistent position generation and probabilistic particle filter algorithms to label potential sources for anonymous distances and to disambiguate the trajectories of the multiple concurrent targets. To avoid conflicts of the close by targets and for reliable initialization, we have also developed a location based concurrent transmission scheduling algorithm. Further work includes more flexible wavefront detection technique to improve threshold based detection which is to further shorten the aftershock and to make the detection be more robust to echoes and noises. 
\bibliographystyle{abbrv}
\bibliography{IPSN}
\section{Apendix}
Parameters in (\ref{equ:blind}) can be expanded as: 
\begin{equation}
    \theta=\arccos \frac{d_{a,b}}{2r}
    \label{equ:areaofchord}
\end{equation}
and 
\begin{equation}
    S_{e}=\int_0^{y_\beta}\left(  2\sqrt{r^2-y^2}-\omega v_u\sqrt{1+\frac{y^2}{d_{a,b}^2-\frac{1}{4}\omega^2 v_u^2}}\ \right)dy
    \label{equ:s3}
\end{equation}
where 
\begin{equation}
    y_\beta=\frac{b_h}{c_h}\sqrt{r^2-b_h^2-2a_hr}\\
    \label{equ:s4}
\end{equation}
refers to the $y$ coordination of intersection point of hyperbola and circle. 
\begin{equation}
    a_h=\frac{v_u\omega}{2}, b_h=\frac{d_{a,b}}{2}, c_h=\sqrt{a_h^2+b_h^2}
    \label{equ:parameter}
\end{equation}
\end{document}